\documentclass[11pt]{article}

\usepackage{etoolbox}

\newbool{stocsubm}
\setbool{stocsubm}{false}

\usepackage[utf8]{inputenc}
\usepackage[T1]{fontenc}
\usepackage[sc]{mathpazo}
\usepackage{microtype}
\usepackage{amsmath,xcolor}
\usepackage{amssymb}
\usepackage{amsthm}
\usepackage{bm,float}
\usepackage{dsfont}
\usepackage{authblk}
\usepackage[margin=20pt,font=small,labelfont=bf]{caption}
\usepackage{fullpage,wrapfig}
\usepackage{mathtools}
\usepackage{thm-restate}
\usepackage[shortlabels,inline]{enumitem}
\usepackage{complexity}
\usepackage{bbm}
\usepackage{xfrac}
\usepackage{subcaption}
\usepackage{xspace}

\usepackage{nag,tikz}
\usetikzlibrary{calc}
\usetikzlibrary{decorations.pathreplacing}
\usetikzlibrary{shapes, patterns, decorations, fit, intersections, arrows, automata, positioning}

\graphicspath{{graphics/}}

\usepackage{algorithm,algpseudocode}
\usepackage[vlined,ruled,linesnumbered,algo2e]{algorithm2e}

\usepackage{multicol}
\usepackage[scaled]{helvet} 
\usepackage{thmtools} 

\usepackage[textsize=footnotesize, color=blue!30!white]{todonotes}
\usepackage[
backend=biber,
style=alphabetic,
citestyle=alphabetic,
maxalphanames=4,
maxcitenames=99,
mincitenames=98,
maxbibnames=99,
giveninits=true,
]{biblatex}

\usepackage{hyperref}
\hypersetup{
    colorlinks=true,
    linkcolor=violet,
    filecolor=magenta,      
    urlcolor=cyan,
    citecolor=blue,
    pdffitwindow=true,
}
\usepackage[capitalize, nameinlink]{cleveref}

\theoremstyle{plain}
\newtheorem{theorem}{Theorem}[section]
\newtheorem{lemma}[theorem]{Lemma}

\newtheorem{conjecture}[theorem]{Conjecture}
\newtheorem{observation}[theorem]{Observation}
\newtheorem{claim}[theorem]{Claim}

\newtheorem{definition}[theorem]{Definition}

\theoremstyle{remark}

\newcommand{\OPTTAP}{\ensuremath{\mathrm{OPT}_{\mathrm{TAP}}}}

\newcommand{\LPOPTkECSS}[1][k]{\ensuremath{\mathrm{LPOPT}_{#1-\mathrm{ECSS}}}}

\newcommand{\OPTkECSS}[1][k]{\ensuremath{\mathrm{OPT}_{#1-\mathrm{ECSS}}}}
\newcommand{\OPTkSECSM}[1][k]{\ensuremath{\mathrm{OPT}_{#1-\mathrm{SECSM}}}}
\newcommand{\LPOPTkSECSM}[1][k]{\ensuremath{\mathrm{LPOPT}_{#1-\mathrm{SECSM}}}}
\newcommand{\LPOPTkECSM}[1][k]{\ensuremath{\mathrm{LPOPT}_{#1-\mathrm{ECSM}}}}
\newcommand{\OPTkECSM}[1][k]{\ensuremath{\mathrm{OPT}_{#1-\mathrm{ECSM}}}}
\newcommand{\LPOPTkGhost}[1][k]{\ensuremath{\mathrm{LPOPT}_{#1-\mathrm{Ghost}}}}

\newcommand{\LPA}{\ensuremath{\mathrm{LP}_{\mathrm{Alg}}}\xspace}

\newcommand{\epsTAP}{\ensuremath{\epsilon_\mathrm{TAP}}}
\newcommand{\gadgetEdges}{\ensuremath{E_{\textrm{Gadget}}}}
\newcommand{\eps}{\ensuremath{\epsilon}}

\newcommand{\Ind}[1]{\ensuremath{\mathbb{1}}\left[#1\right]}

\SetCommentSty{mycommfont}

\DeclareMathOperator{\fracpart}{frac}
\DeclareMathOperator{\supp}{supp}
\DeclareMathOperator{\cov}{cov}

\newcommand\lpecss[1][k]{{\hyperlink{kecsslp}{$#1$-ECSS LP}}}

\title{Ghost Value Augmentation for $k$-Edge-Connectivity\ifbool{stocsubm}{}{\thanks{
This project received funding from Swiss National Science Foundation grant 200021\_184622 and the European Research Council (ERC) under the European Union's Horizon 2020 research and innovation programme (grant agreement No 817750). The second author was supported by the National Science Foundation under Grants No.~DMS-1926686, DGE-1762114, and CCF-1813135.
}}}

\ifbool{stocsubm}{
  \author{}
  \date{}
}{
\author{D Ellis Hershkowitz} 
\affil{Brown University}
\author{Nathan Klein} 
\affil{Institute for Advanced Study}
\author{Rico Zenklusen} 
\affil{ETH Zurich}
}

\begin{document}

\maketitle

\begin{abstract}
    We give a poly-time algorithm for the $k$-edge-connected spanning subgraph ($k$-ECSS) problem that returns a solution of cost no greater than the cheapest $(k+10)$-ECSS on the same graph. Our approach enhances the iterative relaxation framework with a new ingredient, which we call ghost values, that allows for high sparsity in intermediate problems.
    
    Our guarantees improve upon the best-known approximation factor of $2$ for $k$-ECSS whenever the optimal value of $(k+10)$-ECSS is close to that of $k$-ECSS.
    This is a property that holds for the closely related problem $k$-edge-connected spanning multi-subgraph ($k$-ECSM), which is identical to $k$-ECSS except edges can be selected multiple times at the same cost.
    As a consequence, we obtain a $\left(1+O(\sfrac{1}{k})\right)$-approximation algorithm for $k$-ECSM, which resolves a conjecture of Pritchard and improves upon a recent  $\left(1+O(\sfrac{1}{\sqrt{k}})\right)$-approximation algorithm of Karlin, Klein, Oveis Gharan, and Zhang.
    Moreover, we present a matching lower bound for $k$-ECSM, showing that our approximation ratio is tight up to the constant factor in $O(\sfrac{1}{k})$, unless $\P=\NP$.
\end{abstract}

\ifbool{stocsubm}{}{
\begin{tikzpicture}[overlay, remember picture, shift = {(current page.south east)}]
\begin{scope}[shift={(-1.1,2.5)}]
\def\hd{2.5}
\node at (-2.9*\hd,0) {\includegraphics[height=1.3cm]{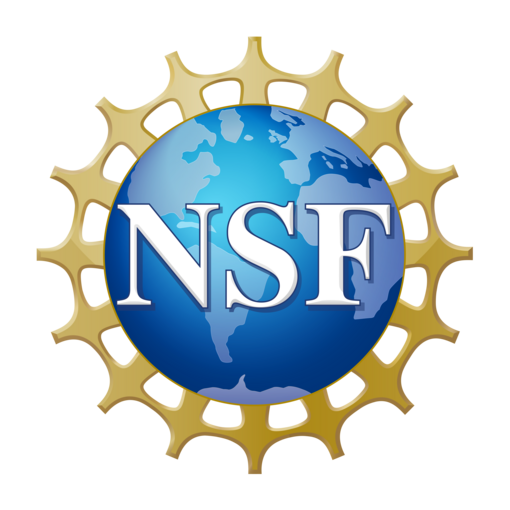}};
\node at (-2*\hd,0) {\includegraphics[height=0.5cm]{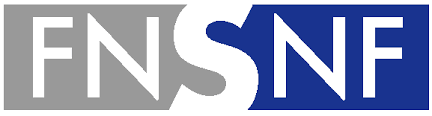}};
\node at (-\hd,0) {\includegraphics[height=1.0cm]{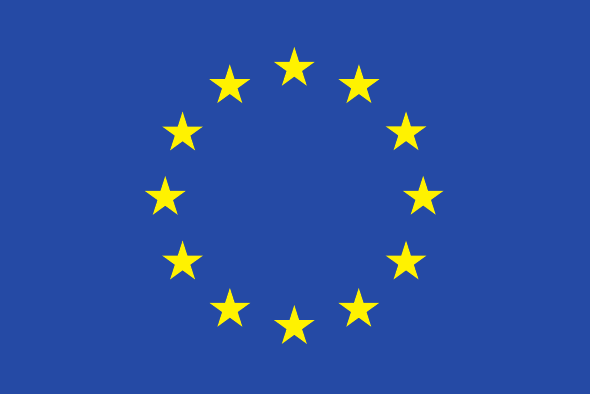}};
\node at (-0.2*\hd,0) {\includegraphics[height=1.2cm]{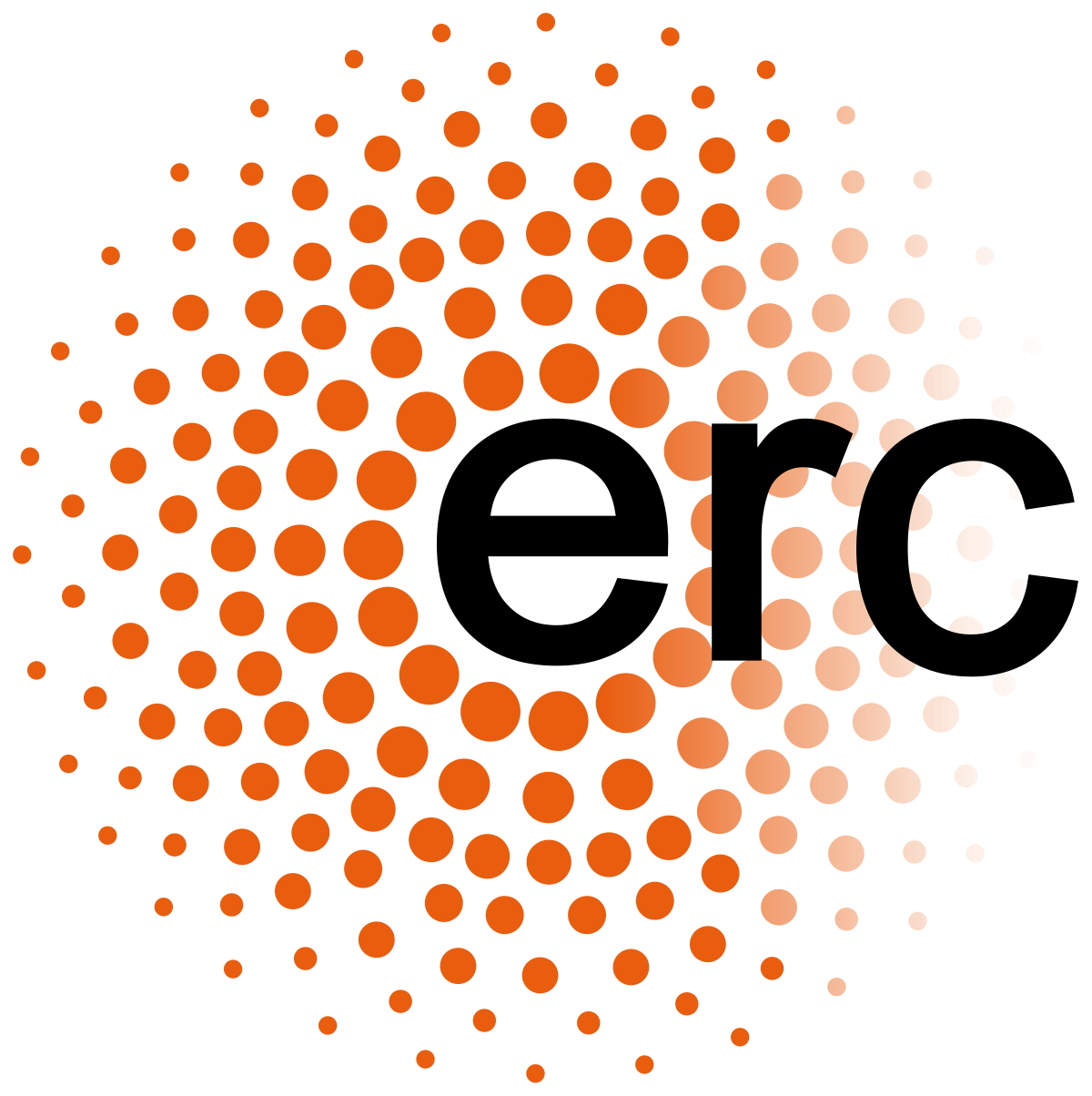}};
\end{scope}
\end{tikzpicture}
}

\thispagestyle{empty}
\newpage
\setcounter{page}{1}

\section{Introduction}

Computing $k$-edge-connected subgraphs of minimum cost is a fundamental problem in combinatorial optimization. For $k=1$, this is the famous minimum spanning tree (MST) problem. For $k \ge 2$, this problem is known as the $k$-edge-connected spanning subgraph ($k$-ECSS) problem. Formally, in $k$-ECSS we are given a multi-graph $G=(V,E)$ with an edge cost function $c:E \to \mathbb{R}_{\ge 0}$ and our goal is to find a set of edges $F \subseteq E$ where $(V, F)$ is $k$-edge-connected while minimizing $c(F) \coloneqq \sum_{e\in F} c(e)$.

For $k\ge 2$, $k$-ECSS is \APX-Hard.
More precisely, there exists an $\epsilon > 0$ such that, unless $\P = \NP$, there is no poly-time $(1+\epsilon)$-approximate algorithm for $k$-ECSS for any $k \ge 2$~\cite{Pri11}. Thus, somewhat surprisingly, the problem does not seem to get easier as $k$ grows. The problem and its special cases have been extensively studied~\cite{GGTW09,dory2018distributed,chalermsook2022approximating,khuller1994biconnectivity,gabow1998efficient,czumaj1999approximability,berger2007minimum,czumaj2004approximation,TZ21,traub2022local,adjiashvili2018beating,fiorini2018approximating,Nut17,iglesias2017coloring,cheriyan2018approximating,GKZ18,CTZ21,grandoni_2022_breaching,traub_2023_approximation}, but for any $k \geq 2$ the best poly-time approximation algorithm is a four-decades-old $2$-approximation due to \textcite{FJ81,FJ82}.\footnote{It also follows from the iterative rounding of \textcite{Jai01}.} Designing a better-than-2 approximation algorithm for any $k \geq 2$ is a major open problem.
Connectivity problems have also been considered for the notion of vertex connectivity (see, e.g.,~\cite{angelidakis_2023_node,hyatt-denesik_2023_finding,auletta_1999_2-approximation,cheriyan_2014_approximating,nutov_2014_approximating}). Notably, there is a $(4+\epsilon)$-approximation for any constant $k$ due to Nutov \cite{nutov_2022_4eapproximation}, improving upon the earlier 6-approximation of Cheriyan and Végh \cite{cheriyan_2014_approximating}.


In this work, we show that one can achieve a close-to-$1$ approximation for $k$-ECSS whenever the cost of the optimal $(k+10)$-edge-connected solution is close to the cost of the optimal $k$-edge-connected solution. Specifically, we give a \textit{resource augmentation} result for $k$-ECSS whereby we compare the quality of our algorithm's output to an adversary that has fewer resources (namely, the resource of cost budget): 
\begin{quote}\centering
    We show that one can poly-time compute a $k$-edge-connected graph with cost no greater than that of the optimal solution which $(k+10)$-edge-connects the graph.
\end{quote}

Similar resource augmentation results are known for other well-studied network design problems such as in the minimum cost bounded degree spanning tree problem~\cite{Goe06,chaudhuri2009would,konemann2000matter,konemann2003primal,ravi2006delegate}; here, the state-of-the-art is an algorithm of \textcite{SL15}, which shows that one can find a spanning tree of maximum degree $d$ that has cost no greater than the spanning tree of maximum degree $d-1$ in poly-time. Likewise, resource augmentation is often considered in online and scheduling algorithms~\cite{ST85, phillips1997optimal, jaillet2008generalized, chekuri2004multi, roughgarden2020resource}. However, to our knowledge, no similar results are known for $k$-connectivity-type problems, and standard techniques~\cite{NW61} would require augmenting the budget by $k$ rather than $O(1)$.

Our result is cost-competitive with the optimal LP solution, defined as follows. For ease of presentation, we fix an arbitrary vertex $r\in V$, called the \emph{root}, and represent each cut $(S, V\setminus S)$ by the side that does not contain $r$. Letting $\delta(S)$ be the set of edges crossing a cut $S \subseteq V$ and $x(F) \coloneqq \sum_{e \in F} x_e$ for $F \subseteq E$, we have the following LP for $k$-ECSS.
\begin{equation}\label{lp:ecss}
\begin{array}{rrcll}
  \min &c^\top x     &     &                       &\\
       &x(\delta(S)) &\geq &k                      &\forall S\subseteq V\setminus \{r\}, S\neq \emptyset \\
       &x            &\in  &{[0,1]}^E &
\end{array}\tag{$k\mathrm{-ECSS}~\mathrm{LP}$}\hypertarget{kecsslp}{}
\end{equation}
Denoting by $\LPOPTkECSS[k]$ the cost of an optimal solution to~\ref{lp:ecss} and by $\OPTkECSS[k]$ the cost of the optimal $k$-ECSS solution, our main result for $k$-ECSS is as follows.

\begin{restatable}{theorem}{mainECSS}\label{thm:mainECSS}
    There is a poly-time algorithm that, for any $k$-ECSS instance with $k\in \mathbb{Z}_{\geq 1}$, returns a $k$-ECSS solution of cost at most $\LPOPTkECSS[(k+10)] \leq \OPTkECSS[(k+10)]$.
\end{restatable}

\noindent Equivalently, we show that one can find in poly-time a $(k-10)$-edge-connected subgraph of cost at most $\OPTkECSS[k]$. Thus, in some sense, our result demonstrates that achieving the last small constant amount of connectivity is the \NP-hard part of $k$-ECSS. 

One might reasonably wonder if it is often the case that the optimal $k$-edge-connected solution has cost close to the optimal $(k+10)$-edge-connected solution. In fact, a well-studied problem closely related to $k$-ECSS called the $k$-edge-connected spanning \textit{multi}-subgraph ($k$-ECSM) problem satisfies exactly this property. $k$-ECSM is the same as $k$-ECSS except our solution $F$ is a multiset that can include each edge of $E$ as many times as we want (where we pay for every copy of an edge). The canonical LP for $k$-ECSM is the same as that of $k$-ECSS but has no upper bound on how many times we choose an edge.
\begin{equation}\label{lp:ecsm}
\begin{array}{rrcll}
  \min &c^\top x     &     &                       &\\
       &x(\delta(S)) &\geq &k                      &\forall S\subseteq V\setminus \{r\}, S\neq \emptyset \\
       &x            &\in  &\mathbb{R}^E_{\geq 0}. &
\end{array}\tag{$k\mathrm{-ECSM}~\mathrm{LP}$}
\end{equation}
We notate by $\LPOPTkECSM[k]$ the cost of an optimal solution to~\ref{lp:ecsm} and by $\OPTkECSM[k]$ the cost of an optimal $k$-ECSM solution. It is easy to see that scaling the optimal $k$-ECSM LP solution by $(k+10)/k$ results in a $(k+10)$-ECSM LP solution, so we have the claimed relation between the costs of the optimal $k$-edge-connected and $(k+10)$-edge-connected LP solutions:
\begin{align}\label{eq:scaling}
    \LPOPTkECSM[(k+10)] \leq \left(1 + \frac{10}{k} \right) \cdot \LPOPTkECSM[k].
\end{align}

Unlike $k$-ECSS, it is known that, as $k \to \infty$, approximation factors arbitrarily close to $1$ are possible; however, prior to this work, the correct asymptotic dependence on $k$ was not fully understood. \textcite{GGTW09} showed that \emph{if the graph $G$ is unweighted} (i.e., every edge has cost $1$), then $k$-ECSM (and $k$-ECSS) admits $(1+\sfrac{2}{k})$-approximation algorithms. This led \citeauthor{Pri11} to pose the following conjecture for $k$-ECSM (with general weights).
 \begin{conjecture}[\cite{Pri11}]\label{conj:kECSM}
   $k$-ECSM admits a poly-time $(1+O(\sfrac{1}{k}))$-approximation algorithm.
 \end{conjecture}
 \noindent Very recently, \textcite{KKOZ22} made significant progress on this conjecture, showing there is a poly-time $(1+\sfrac{5.06}{\sqrt{k}})$-approximation for $k$-ECSM. However, their approach provably does not give better than a $(1 + O(\sfrac{1}{\sqrt{k}}))$-approximation. 
 
 The result of \textcite{KKOZ22} added to a considerable body of work on $k$-ECSM. The first notable algorithm for $k$-ECSM is due to \textcite{FJ81,FJ82}, who gave a $\sfrac{3}{2}$-approximation algorithm for even $k$ and a $(\sfrac{3}{2} + O(\sfrac{1}{k}))$ for odd $k$. This algorithm essentially follows by a reduction to the well-known Christofides-Serdyukov algorithm for the traveling salesperson problem (TSP). Despite many subsequent works on $k$-ECSM~\cite{JT00,KR96, Kar99,Gab05,GG08,Pri11,LOS12,CR98,BFS16,SV14,BCCGISW20}, this algorithm of \textcite{FJ81,FJ82} remained the best approximation algorithm for nearly four decades, except when the underlying graph is unweighted, $k\gg \log n$, or $k=2$. Recently, this was very slightly improved to a $(\sfrac{3}{2}-\epsilon)$-approximation for even $k$ where $\epsilon = 10^{-36}$~\cite{KKO21, KKO22}. Thus, for large $k$, the algorithm of \textcite{KKOZ22} considerably improved on all known prior work for $k$-ECSM and gave the first algorithm with approximation ratio tending to 1 as $k \to \infty$ independent of $n$.
 


As an immediate consequence of \cref{thm:mainECSS} and \Cref{eq:scaling}, we are able to settle \citeauthor{Pri11}'s conjecture with a poly-time $(1 + O(\sfrac{1}{k}))$-approximation for $k$-ECSM.

%
%
\begin{restatable}{theorem}{mainECSM}\label{thm:mainECSM}
  There is a poly-time algorithm for $k$-ECSM that, for any $k$-ECSM instance with $k\in \mathbb{Z}_{\geq 1}$, returns a $k$-ECSM solution of cost at most $(1+\frac{10}{k}) \cdot \LPOPTkECSM[k] \leq (1+\frac{10}{k}) \cdot \OPTkECSM[k]$.
\end{restatable}
\begin{proof}[Proof of \Cref{thm:mainECSM} assuming $k$ is given in unary] We give an algorithm which runs in poly-time assuming $k$ is represented in unary in the input. See \hyperref[pf:mainECSMStrongPoly]{this proof} in \Cref{sec:rounding} for a more general proof in which we assume $k$ is represented in binary.

Suppose our instance of ECSM is on graph $G= (V,E)$ with costs $c$, and optimal $k$ and $k+10$ LP costs $\LPOPTkECSM[k]$ and $\LPOPTkECSM[(k+10)]$, respectively. Then, consider the $(k+10)$-ECSS instance on $G' = (V, E')$ with costs $c$ and optimal LP cost $\LPOPTkECSS[(k+10)]$, where $E'$ is $E$ with $k+10$ copies of each $e \in E$.

  Apply \Cref{thm:mainECSS} to our $(k+10)$-ECSS instance to compute a multiset $F \subseteq E'$ which $k$-edge-connects $V$. By construction, $F$ is a feasible solution to our $k$-ECSM instance, when interpreting the selection of parallel edges by one multi-selection of the edge in $E$ they correspond to (and any $k$-ECSM solution on $G$ can be interpreted as a solution to $k$-ECSS on $G'$), and by \Cref{eq:scaling} it has cost at most
\begin{align*}
    \LPOPTkECSS[(k+10)] \leq \LPOPTkECSM[(k+10)] \leq \left(1 + \frac{10}{k} \right) \cdot \LPOPTkECSM[k]. \qquad \qedhere
\end{align*}
\end{proof}
\noindent Our algorithm improves upon the $1+\sfrac{5.06}{\sqrt{k}}$ algorithm of~\cite{KKOZ22} for all relevant $k$ (i.e., those in which that algorithm was better than $\sfrac{3}{2} - \eps$).

As previously noted \cite{GB93,Pri11}, a poly-time algorithm for $k$-ECSM whose solutions are approximate with respect to \LPOPTkECSM[k] gives an approximation algorithm with the same approximation bounds for subset (a.k.a.\ Steiner) $k$-ECSM. Here, we only need to $k$-edge-connect a subset of the nodes. Thus, as a consequence of \Cref{thm:mainECSM}, we get a $1+\sfrac{10}{k}$ for this more general subset $k$-ECSM problem. See \Cref{sec:steiner} for a formal statement of this problem and result.


Complementing this solution to Pritchard's conjecture, we show that the dependence on $k$ in our algorithm is essentially optimal. Specifically, we show that \cref{thm:mainECSM} is almost tight by showing $(1 + \Omega(\sfrac{1}{k}))$-hardness-of-approximation for $k$-ECSM. (The tightness is up to the constant in front of the term $\sfrac{1}{k}$.)

\begin{restatable}{theorem}{hardness}\label{thm:hardness}
There exists a constant $\epsTAP > 0$ (given by \Cref{thm:tapAPXHard}) such that there does not exist a poly-time algorithm which, given an instance of (unweighted) $k$-ECSM where $k$ is part of the input, always returns a $\left(1+\frac{\epsTAP}{9k}\right)$-approximate solution, unless $\P = \NP$.
\end{restatable}
\noindent Such a hardness result was also identified as an open question by \textcite{Pri11}.



We first give a proof sketch in \cref{subsec:proofsketch}. In \cref{sec:rounding}, we then describe our main technical rounding theorem and how it implies our results on $k$-ECSS and $k$-ECSM. In particular, we show it allows us to compute a $k$-ECSS solution with cost at most that of the optimal $(k+10)$-ECSS solution, and we also show it implies a $1 + O(\sfrac{1}{k})$-approximate $k$-ECSM algorithm---settling Pritchard's $k$-ECSM conjecture. Later, in \cref{sec:hardness}, we show a matching $1 + \Omega(\sfrac{1}{k})$ hardness of approximation for $k$-ECSM. In \cref{sec:alg}, we describe our algorithm and give further intuition on why it gives strong guarantees before formally analyzing its performance in \cref{sec:proof}.


\subsection{Iterative Relaxation Barriers and Ghost Value Augmentations to Overcome Them}\label{subsec:proofsketch}

First note that for the rest of the paper, to slightly simplify notation, we will work with a solution to the \lpecss[k] and round it to a $(k-10)$-edge-connected graph. This is of course equivalent to working with $k$-edge-connectivity and the \lpecss[(k+10)].

Our approach to showing \Cref{thm:mainECSS} is to apply iterative LP relaxation methods to the \lpecss[k] and obtain a $(k-10)$-edge-connected graph. Iterative LP methods are a well-studied approach first pioneered by \textcite{Jai01}. See \textcite{LRS11} for a comprehensive overview. Generally, iterative relaxation algorithms repeatedly compute an optimal solution to a suitable LP and then make progress towards producing a desired output by either finding an edge with LP value $0$ that can be deleted, a newly integral edge that can be frozen at its integral value, or some LP constraint which is \emph{nearly} satisfied and can be dropped from future LP recomputations while approximately preserving feasibility.



Applying such an iterative relaxation approach to the \lpecss[k] is naturally suited to proving \Cref{thm:mainECSS} since we are allowed $O(1)$ slack in connectivity. Specifically, one might hope to argue that if no edge can be frozen then there is some constraint of the \lpecss[k] corresponding to a cut which has at least $k-O(1)$ frozen edges crossing it. Such a constraint can be safely dropped from future recomputations since our ultimate goal allows for slack $O(1)$ in connectivity. However, it is not too hard to see that such a natural approach faces a significant barrier and so a non-standard idea is required.

In what follows, we describe this approach and barrier in more detail and how our key non-standard idea of ``ghost value augmentations''  allows us to overcome this barrier.

\paragraph{A Standard Iterative Relaxation Approach.} A first attempt to prove \Cref{thm:mainECSS} is as follows, where we replace the constant 10 with an arbitrary constant $c \in \mathbb{Z}_{\ge 1}$. Repeat the following until there are no remaining variables.

\begin{quote}
	 Let $y$ be an extreme point solution to \lpecss[k]. For each edge $e$ with $y_e = 0$, delete $e$ from the LP. For each edge $e$ with $y_e = 1$, add the constraint $x_e = 1$ to the LP and call $e$ ``frozen.'' We let $F$ be all edges we have frozen so far.
	
	Now suppose that whenever we cannot delete or freeze a new edge, there is always a cut $S \subseteq V \setminus \{r\}$ such that $|\delta(S) \cap F| \ge k-c$. We call this the \textit{light cut property}.  In this case, we can delete the constraint $x(\delta(S)) \ge k$ from the LP. This is a safe operation since $\delta(S)$ already has at least $k-c$ frozen edges. 
\end{quote}	
 Once there are no remaining variables, we have an integral solution and can return the set of frozen edges. Provided one can still efficiently solve the LP even after dropping constraints and so long as the light cut property always holds when we cannot delete or freeze a new edge, standard arguments would demonstrate that the above algorithm always returns an integral $(k-c)$-ECSS solution of cost no more than the cost of $\LPOPTkECSS[k]$. 

\paragraph{A Barrier to the Standard Approach.} 
We do not know if the light cut property is true or not, and proving or disproving it is a very interesting open problem. However, there is a major barrier to proving it using known techniques for iterative relaxation, which we detail here. 

To prove results like the light cut property, one generally first demonstrates that the set of tight constraints at any stage of the algorithm's execution can be ``uncrossed'' to obtain a laminar family.\footnote{A family of sets ${\cal L}$ is called \textit{laminar} if it does not contain any pair of intersecting sets $S,T$, which are sets such that $S\cap T \not\in \{S, T, \emptyset\}$.
Moreover, two sets $S, T \subseteq V$ are \textit{crossing} if all of $S \cap T$, $V \setminus (S \cup T)$, $S \setminus T$, and $T \setminus S$ are non-empty.
We typically use these notions for vertex sets that correspond to cuts.
Because we assume that cuts do not contain $r$, the notions of crossing and intersecting sets coincide for cuts.
} 
In our case, a set $S \subseteq V$ is tight if $y(\delta(S)) = k$ and one wants to show that there is a laminar family ${\cal L} \subseteq 2^V$ such that every constraint $x(\delta(S)) \geq k$ where $S$ is tight is spanned by the constraints $\{ x(\delta(S)) \geq k \colon S \in {\cal L}\}$.\footnote{Here, and throughout this work, we treat $x$ as a  variable and $y$ as a fixed solution to our LP.}

The major barrier to the standard iterative relaxation approach is that it does not appear that uncrossing is possible. One sufficient criteria for uncrossing is if the constraints can be expressed as $x(\delta(S)) \geq f(S)$ where the ``requirement function''  $f:V \to \mathbb{Z}_{\ge 0}$ is skew supermodular~\cite{Jai01}. 

A function $f:V \to \mathbb{Z}_{\ge 0}$ is called skew supermodular if for all $S,T \subseteq V$ we have either 
\begin{align*}
  f(S) + f(T) &\le f(S \cap T) + f(S \cup T), \text{ or}\\
  f(S) + f(T) &\le f(S \smallsetminus T) + f(T \smallsetminus S).
\end{align*}	

At the beginning of our process, $f(S) = k$ for all $S \subseteq V$, and thus both of these inequalities trivially hold with equality. However, once we drop a constraint $S$ for which we had $|F \cap \delta(S)| \ge k-c$, this property fails to hold. In particular, we may have dropped the constraint for $S \cap T$ and $S \smallsetminus T$ so that $f(S \cap T) = f(S \smallsetminus T) = 0$. Thus, in this situation, the left-hand side of each equation would be $2k$ and the right-hand side $k$. For example, the situation in \Cref{fig:skewsupermodularissue} could occur, where we dropped cuts with at least $k-c$ frozen edges.

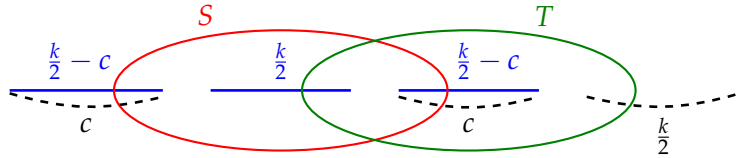
\begin{figure}[htb!]
	\centering 
  \begin{tikzpicture}[node distance=2.5cm, line width=1pt]
  \begin{scope}[every node/.style={inner sep=3mm}]
  \node (a) {};
  \node (b) [right of=a] {};
  \node (c) [right of=b] {};
  \end{scope}
  \node (invis-right) [right of=c,draw=none] {};
  \node (invis-left) [left of=a,draw=none] {};
  \node (invis-up) [above of=b,draw=none] {};

  \begin{scope}[-,blue]
    \draw (invis-left) -- (a) node[midway,above,xshift=-0.1cm] {\( \frac{k}{2}-c \)};
    \draw (a) -- (b) node[midway,above] {\( \frac{k}{2} \)};
    \draw (b) -- (c) node[midway,above,xshift=0.25cm] {\( \frac{k}{2}-c \)};
  \end{scope}

  \begin{scope}[-,dashed, bend right=15]
    \draw (invis-left) to node[below] {$c$} (a);
    \draw (b) to node[below] {$c$} (c);
    \draw (c) to node[below] {$\frac{k}{2}$} (invis-right);
  \end{scope}

  \node[ellipse, thick, red, draw, fit=(a) (b), inner sep=0pt, minimum height=16mm, label={[xshift=-1.0cm, yshift=-0.1cm, red]$S$}] {};
  \node[ellipse, thick, green!50!black, draw, fit=(b) (c), inner sep=0pt, minimum height=16mm,label={[xshift=1.0cm, yshift=-0.1cm, green!50!black]$T$}] {};
\end{tikzpicture}
	\caption{An example where skew supermodularity of the requirement function fails. The solid blue edges represent collections of frozen edges between sets. The dashed edges represent collections of non-frozen edges. Here, $S \smallsetminus T$ and $S \cap T$ have been dropped as they have at least $k-c$ frozen edges. However, $S,T,T \smallsetminus S$, and $S \cup T$ have fewer than $k-c$ frozen edges and thus remain. So, $f(S)=f(T)=k$ but $f(S \smallsetminus T)=f(S \cap T)=0$.}\label{fig:skewsupermodularissue}
\end{figure}

\Cref{fig:skewsupermodularissue} shows that we cannot apply the result of Jain as a black box. However, one can still successfully uncross in this situation. The reason is that the constraints $S \smallsetminus T, T \smallsetminus S, S \cap T, S \cup T$ are all still tight. Therefore, even though the constraints are not present in the LP, one can still replace the constraints $S$ and $T$ with the constraints $S \smallsetminus T$ and $T \smallsetminus S$ (or $S \cap T$ and $S \cup T$). 

The true issue arises when the connectivity of some sets drop below $k$, as in \Cref{fig:realissue}. In \Cref{fig:realissue}, \textit{none of} $S \smallsetminus T$, $T \smallsetminus S$, $S \cap T$, and $S \cup T$ are tight constraints. One could still consider adding them to the family. However, if we did this, we would face two major issues. 
\begin{enumerate}
    \item The constraints may not be integers, as is the case for all of the sets here. This leads to problems in the next phase of the iterative relaxation argument which uses integrality of the constraints to argue that if $S,T$ are two tight sets in the family then the symmetric difference $\delta(S) \Delta \delta(T)$ has size at least 2.
    \item  Unlike in the case in which $S$ and $T$ are minimum cuts of the graph (in which we can apply standard uncrossing) it is not necessarily the case that 
      $$\chi^{\delta(S)} + \chi^{\delta(T)} = \chi^{\delta(S \setminus T)} + \chi^{\delta(T \setminus S)},$$
    where $\chi^F$ for $F \subseteq E$ is the vector in ${\{0,1\}}^E$ that is 1 at the edges in $F$ and 0 elsewhere.
    In this example this equality does not hold due to the edge with value $a$. 
    Similarly, it is not necessarily the case that $$\chi^{\delta(S)} + \chi^{\delta(T)} = \chi^{\delta(S \cap T)} + \chi^{\delta(S \cup T)}.$$
To see this, one can extend the example by adding a small fractional edge between $S \smallsetminus T$ to $T \smallsetminus S$ and adjusting other edges accordingly (since only $S,T$ are tight, this is not difficult).

However, relations as the two highlighted above are central in classical uncrossing arguments.
\end{enumerate} 
Therefore, to prove the light cut property, one would likely have to deal with a fairly complicated family of tight sets.

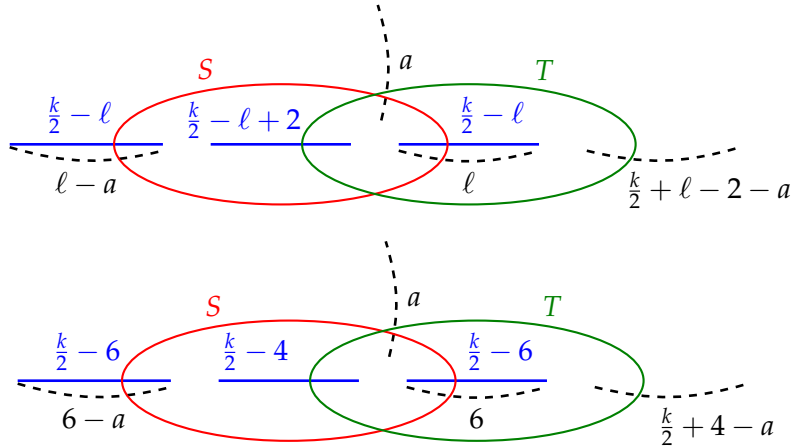
\begin{figure}[htb!]
\centering 
  \begin{tikzpicture}[node distance=2.5cm, line width=1pt]
  \begin{scope}[every node/.style={inner sep=3mm}]
   \node (a) {};
   \node (b) [right of=a] {};
   \node (c) [right of=b] {};
  \end{scope}
   \node (invis-right) [right of=c,draw=none] {};
   \node (invis-left) [left of=a,draw=none] {};
   \node (invis-up) [above of=b,yshift=-5mm,draw=none] {};

   \begin{scope}[-,blue]
     \draw (invis-left) -- (a) node[midway,above,xshift=-0.1cm] {$\frac{k}{2}-\ell$};
     \draw (a) -- (b) node[midway,above,xshift=-5.0mm,yshift=-0.7mm] {$\frac{k}{2}-\ell+2$};
     \draw (b) -- (c) node[midway,above,xshift=0.3cm] {$\frac{k}{2}-\ell$};
   \end{scope}

   \begin{scope}[-,dashed,bend right=15]
     \draw (invis-left) to node[below] {$\ell-a$} (a);
     \draw (b) to node[below] {$\ell$} (c);
     \draw (c) to node[below,xshift=6mm] {$\frac{k}{2}+\ell-2-a$} (invis-right);
     \draw (b) to node[right] {\(a\)} (invis-up);
   \end{scope}

  \node[ellipse, thick, red, draw, fit=(a) (b), inner sep=0pt, minimum height=16mm, label={[xshift=-1.0cm, yshift=-0.1cm, red]$S$}] {};
  \node[ellipse, thick, green!50!black, draw, fit=(b) (c), inner sep=0pt, minimum height=16mm,label={[xshift=1.0cm, yshift=-0.1cm, green!50!black]$T$}] {};
\end{tikzpicture}
  \begin{tikzpicture}[node distance=2.5cm, line width=1pt]
  \begin{scope}[every node/.style={inner sep=3mm}]
   \node (a) {};
   \node (b) [right of=a] {};
   \node (c) [right of=b] {};
  \end{scope}
   \node (invis-right) [right of=c,draw=none] {};
   \node (invis-left) [left of=a,draw=none] {};
   \node (invis-up) [above of=b,yshift=-5mm,draw=none] {};

   \begin{scope}[-,blue]
     \draw (invis-left) -- (a) node[midway,above,xshift=-0.1cm] {$\frac{k}{2}-6$};
     \draw (a) -- (b) node[midway,above,xshift=-0.45cm] {$\frac{k}{2} - 4$};
     \draw (b) -- (c) node[midway,above,xshift=0.3cm] {$\frac{k}{2} - 6$};
   \end{scope}

   \begin{scope}[-,dashed,bend right=15]
     \draw (invis-left) to node[below] {$6-a$} (a);
     \draw (b) to node[below] {$6$} (c);
     \draw (c) to node[below,xshift=6mm] {$\frac{k}{2}+4-a$} (invis-right);
     \draw (b) to node[right] {\(a\)} (invis-up);
   \end{scope}

  \node[ellipse, thick, red, draw, fit=(a) (b), inner sep=0pt, minimum height=16mm, label={[xshift=-1.0cm, yshift=-0.1cm, red]$S$}] {};
  \node[ellipse, thick, green!50!black, draw, fit=(b) (c), inner sep=0pt, minimum height=16mm,label={[xshift=1.0cm, yshift=-0.1cm, green!50!black]$T$}] {};
\end{tikzpicture}
\caption{An example of where uncrossing breaks down, where in the top figure we let $\ell=\lceil \frac{c+1}{2} \rceil$ and in the bottom figure we assume $c=10$ for concreteness. The blue edges represent collections of frozen edges and the dotted edges represent collections of fractional edges. To ensure feasibility we let $0 < a < 1$. In this example, we drop cuts when there are $k-c$ frozen edges (or $k-10$ in the bottom picture). Then, in both examples we have dropped $S \smallsetminus T$ and $S \cap T$, but we have not yet dropped $S$ and $T$.}\label{fig:realissue}
\end{figure}

\paragraph{Overcoming the Barrier with Ghost Values.} 
Instead of working with this uncrossable family, we introduce a relaxation approach we call ``ghost value augmentation.'' We consider the LP solution $y$ together with a ghost vector $g$ that augments $y$, so that the LP constraints are now of the form $x(\delta(S)) + g(\delta(S)) = (x+g)(\delta(S)) \ge k$. We say such a constraint is tight with respect to solution $y$ if $(y+g)(\delta(S)) = k$. This ghost vector will help us achieve uncrossing of tight sets. However, crucially, it will never be used in the final solution. This ensures that we never increase the cost of the solution compared to the LP.

We will still follow the general framework of iterative relaxation. Given an extreme point solution $y$, we will delete edges with $y_e = 0$ and freeze edges with $y_e = 1$. And, as before, we will drop constraints corresponding to tight sets $S$ with the property that $|\delta(S) \cap F| \ge k-O(1)$ (where $F$ is the set of frozen edges). However, we will only drop such a set $S$ if:
%
\begin{enumerate}[(i)]
  \item\label{item:dropcondLPConstr} \textbf{$S$ is a minimal tight set corresponding to an LP constraint}. We restrict ourselves to tight sets for which there is no $T \subsetneq S$ such that the cut constraint corresponding to $T$ is tight and still in the LP. Such sets $S$ are desirable because they have the additional property that they are either vertices or all edges with both endpoints inside them are frozen. 
  \item\label{item:dropcondLowFrac} \textbf{$\delta(S)$ has only $O(1)$ fractional edges}. This allows us to ensure that the cut $\delta(S)$ does not change much over the course of the remainder of the algorithm. In particular, this will let us derive \textit{upper bounds} on the number of edges crossing a dropped set, which in turn helps to lower bound the value of other cuts.
\end{enumerate}
Of course, point \ref{item:dropcondLowFrac} implies that there are $k-O(1)$ frozen edges, so it is sufficient to check \ref{item:dropcondLPConstr} and \ref{item:dropcondLowFrac} for all sets. These points together allow us to argue that after dropping such a constraint $S$ with $|S| \ge 2$, it is safe to contract $S$ to a vertex.\footnote{Contracting a vertex set $S\subseteq V$ in $G=(V,E)$ means that we remove $S$ from $V$ and add a single new vertex $S$.
An edge $(u,v)\in E$ before contraction with $u\in S$ and $v\in V\setminus S$ will become an edge between $S$ and $v$; moreover, edges with both endpoints in $S$ will be removed in the contracted graph.
As is common, when having an edge $e=(u,v)\in E$ before contraction that gets transformed into an edge $(S,v)$ after contraction, we still consider this to be the same edge.
In particular, any LP value or constraint on $e$ before contraction will be interpreted as a value or constraint, respectively, on the new edge after contraction.} This is because cuts contained in $S$ will not change by more than $O(1)$ throughout the execution of the entire algorithm, as they only contain edges inside $S$, which are all frozen by \ref{item:dropcondLPConstr}, and edges in $\delta(S)$, of which all but $O(1)$ are frozen by \ref{item:dropcondLowFrac}. See \Cref{fig:contract}.

\begin{figure}[htb!]
    \centering 
    \begin{tikzpicture}[line width=1pt,xscale=1.2,yscale=0.8]
  \begin{scope}[every node/.style={draw,circle,fill=white, inner sep=1mm}]
    \node (A) at (0,2) {};
    \node (B) at (2,2) {};
    \node (C) at (2,0) {};
    \node[fill=blue] (D) at (0,0) {};
  \end{scope}

  \begin{scope}[blue]
    \draw (A) -- node[above] {$\frac{k}{2} + 5$} (B);
    \draw (B) -- node[left] {$\frac{k}{2} + 5$} (C);
    \draw (C) -- node[below] {$\frac{k}{2} + 5$} (D);
  \end{scope}

    \node[draw=red,circle,fit=(A) (B) (C) (D),inner sep=0.2mm] {};

    \begin{scope}[blue,-]
      \draw (A) -- node[above,xshift=-2mm] {$\frac{k}{2}+5$} (-1.6,2);
      \draw (D) -- node[below,xshift=-2mm] {$\frac{k}{2}-4$} (-1.6,0);
    \end{scope}

    \begin{scope}[dashed,-,bend left=15]
      \draw (D) to ++(-1.0,-1.0);
      \draw (D) to ++(0,-1.0);
    \end{scope}

\end{tikzpicture}
\caption{The blue edges are frozen, and the dotted edges are fractional. Consider the red set $S$ with $2$ incident fractional edges. Since  $\delta(S)$ has only $2$ fractional edges, and the edges inside $S$ are all frozen, any set contained in $S$ is already safe, as is the case with the blue vertex.}\label{fig:contract}
\end{figure}
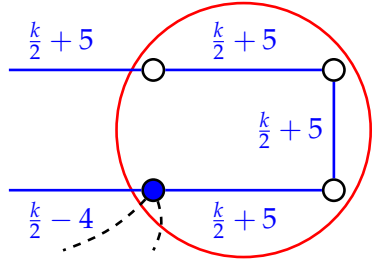

Contraction now gives us the crucial property that at every iteration of the algorithm, operating on a graph $G=(V,E)$ resulting from contracting some number of sets, we have $y(\delta(S)) \geq k$ for all $S \subseteq V \smallsetminus \{r\}$ for which $2 \le |S|$. For all vertices $v \in V$, we only have the weaker guarantee that $y(\delta(v)) \ge k-O(1)$ where we use the notation $\delta(v) = \delta(\{v\})$.
Even though the connectivity of the graph is not uniform, the fact that the only cuts below $k$ are the vertices allows us to show that either:
\begin{enumerate}[(a)]
  \item\label{item:canDoGhostValueAugm} There is an edge $e=(u,v)$ with $(y+g)(E(u,v)) \in [\frac{k}{2}-O(1),\frac{k}{2})$ (where $E(u,v)$ is the set of edges between $u$ and $v$). In this case, we perform a \textit{ghost value augmentation}: we artificially increase the value of $(y+g)(E(u,v))$ to at least $\frac{k}{2}$ by increasing $g_e$ by $O(1)$ for an edge in $E(u,v)$.
  \item\label{item:canUncross} Otherwise, the set of tight constraints can be successfully uncrossed. In this case, we argue that there must be a set $S$ that is safe to drop that has no tight children and at most $3$ fractional edges. If $S$ is not a singleton, then we additionally contract it.
\end{enumerate}
At first, \ref{item:canDoGhostValueAugm} may look like a strange property to expect for $k$-ECSS solutions. Indeed, the underlying input graph may not have any multi-edges, so that at the first iteration $y(E(u,v)) \le 1$ for all vertices $u,v$. However, as sets are contracted, these structures may begin to emerge as barriers for uncrossing. For example, in \cref{fig:realissue}, since $S \smallsetminus T$ and $S \cap T$ have been dropped, they are now singletons. This allows us to study $y(E(u,v))$ for $\{u\}=S \smallsetminus T$ and $\{v\} = S \cap T$. And one can see that $y(E(u,v)) = \frac{k}{2}-O(1)$, giving us a candidate for ghost value augmentation.\footnote{Note that \cref{fig:realissue} is not quite representative of the situations we will arrive at over the course of the algorithm since it was designed to handle the situation in which we drop all constraints with $k-O(1)$ frozen edges and not $O(1)$ fractional edges. However, it is quite similar to situations we study in the following sections. See \cref{fig:lowfracissue} for an instance tailored to our algorithm which shows the importance of ghost value augmentation.} For more intuition about why these structures should appear, one can study the cactus representation of minimum cuts (see~\cite{FF09} for a nice overview of this representation).

In both cases~\ref{item:canDoGhostValueAugm} and~\ref{item:canUncross} we make progress: we either drop a constraint that is safe to drop (and possibly contract its corresponding set), or we fix an edge in our current graph to at least $\frac{k}{2}$. We describe this process in more detail in \cref{sec:alg}. In \cref{sec:proof} we show that one of these two cases must occur and that at the end of the algorithm the frozen edges $(k-O(1))$-connect the graph.

\section{Main Rounding Theorem}\label{sec:rounding}


Our main result will follow immediately from a general rounding theorem. 

\begin{theorem}[Main Rounding Theorem]\label{thm:mainaux}
  There is a poly-time algorithm that, given $y \in \mathbb{R}^E_{\geq 0}$ where $y(\delta(S)) \geq k$ for all non-empty $S \subsetneq V$, returns a $z \in \mathbb{Z}_{\geq 0}^E$ that is:
  \begin{enumerate}[nosep]
    \item \textbf{Cost-Preserved: } $c^{\top} z \leq c^{\top} y$;
      \item \textbf{Integrally-Rounded: }$z_e \in \{\lfloor y_e \rfloor, \lceil y_e\rceil\}$;
      \item \textbf{Highly-Connected: } $z(\delta(S)) \geq k-9-\Ind{k \text{ odd}}$ for all non-empty $S \subsetneq V$.\footnote{$\Ind{k \text{ odd}}$ is $1$ if $k$ is odd and $0$ otherwise.}
  \end{enumerate}
\end{theorem}

We now observe that both our main result for $k$-ECSS (\Cref{thm:mainECSS}) and our main result for $k$-ECSM in polynomial time (\Cref{thm:mainECSM}) are immediate from this rounding theorem.
\mainECSS*
\begin{proof}
    The result is immediate from applying \Cref{thm:mainaux} to any optimal solution $y$ to \lpecss[(k+10)]. $y$ is poly-time computable by the Ellipsoid Method~\cite{groetschel_1981_ellipsoid} and standard separation oracles.
\end{proof}

We note that in fact the same proof shows we can get the same guarantee for a slightly more general problem than $k$-ECSS. In particular, edges can be given arbitrary lower and upper bounds, and we can still obtain a solution in polynomial time (this includes the case in which $k$ and possibly the lower and upper bounds are exponentially large compared to the input size). 

\mainECSM*
\begin{proof}[Proof of \Cref{thm:mainECSM} without assuming that $k$ is given in unary]\phantomsection\label{pf:mainECSMStrongPoly} Let $y$ be an optimal solution to \lpecss[(k+10)] of cost at most $\LPOPTkECSM[(k+10)]$. Next, apply \Cref{thm:mainaux} to $y$ to compute a solution $z$. Return (the edge multiset naturally corresponding to) $z$ as our solution. $y$ is computable by the Ellipsoid Method~\cite{groetschel_1981_ellipsoid} and standard separation oracles. Furthermore, by the guarantees of \Cref{thm:mainaux} we have that our solution is feasible for $k$-ECSM and computable in poly-time. Lastly, by \Cref{thm:mainaux} and \Cref{eq:scaling} its cost is upper bounded by
\begin{align*}
    \LPOPTkECSM[(k+10)] \leq \left(1 + \frac{10}{k} \right) \cdot \LPOPTkECSM[k]. \qquad \qedhere
\end{align*}


\end{proof}



\noindent Thus, in what follows, we focus on showing \cref{thm:mainaux}. Furthermore, observe that in order to do so it suffices to show the result for even $k$ since, if we are given an odd $k$, applying the result to (the even number) $k-1$ immediately gives the result for (the odd number) $k$. Thus, in what follows we assume that $k$ is even.

\section{Algorithm for the Rounding Theorem}\label{sec:alg}
Having reduced our problem to showing \Cref{thm:mainaux}, we proceed to describe the algorithm for \Cref{thm:mainaux}. 

As discussed earlier, our algorithm for \Cref{thm:mainaux} uses iterative relaxation techniques. Informally our algorithm is as follows. We repeatedly solve an LP that tries to round $y$. Each time we solve our LP, either our solution is integral in which case we return our solution, it has a newly integral edge which we freeze at its current value, it has an edge with value $0$ which we delete, or we perform one of the following two relaxations of our LP:
\begin{enumerate}
    \item \textbf{Ghost Value Augmentation:} we costlessly increase the LP value between two carefully chosen vertices. Specifically, if there are two vertices $u$ and $v$ that have nearly $\sfrac{k}{2}$ total LP value on edges between them (namely total edge value in $\left[\sfrac{k}{2}-2, \sfrac{k}{2}\right)$) then we add ``ghost values'' by increasing the LP value between $u$ and $v$ by $2$ at cost $0$.
    \item \textbf{Drop/Contract:} we drop a carefully chosen set's constraint and contract it. Specifically, if there is a tight set $S$ (i.e., a set with exactly $k$ total LP edge mass leaving it) corresponding to a constraint in the LP, such that $S$ is
      \begin{enumerate*}[label= (\roman*)]
    \item minimal in the sense that it contains no strict subset corresponding to a tight LP constraint, and
    \item has at most $3$ fractional edges incident to it,
  \end{enumerate*}
    then we remove the constraint corresponding to $S$ from our LP and (if the set contains $2$ or more vertices) we contract it into a single vertex.
\end{enumerate}
\noindent As we later argue, we will always be in one of the above cases. The final solution we return will ignore the ghost values, and we will show that even after ignoring the ghost values our solution is well-connected.

We now more precisely define our LP given an input vector $y \in \mathbb{R}^E_{\geq 0}$ which we would like to round.  We denote by $\lfloor y \rfloor$ and $\lceil y \rceil$ the integral vector obtained from $y$ by taking the floor and ceiling of each component respectively. We let $x \in [\lfloor y \rfloor, \lceil y \rceil]$ denote that $\lfloor y_e \rfloor \leq x_e \leq \lceil y_e \rceil$ for each $e \in E$. To achieve a clean and elegant bookkeeping of our ghost values, we will save them in a separate vector $g\in \mathbb{Z}^E_{\geq 0}$. Thus, for a given $y$ and $g$, the LP we will solve is the following (potentially with extra constraints for  frozen coordinates).
\begin{equation}\label{lp:ghost}
\begin{array}{rrcll}
  \min &c^\top x     &     &                       &                                                     \\
       &x(\delta(S)) &\geq &k - g(\delta(S))       &\forall S\subseteq V\setminus \{r\}, S\neq \emptyset \\
       &x            &\in  & [\lfloor y \rfloor, \lceil y \rceil]. &\tag{$k\mathrm{-EC~Ghost}~\mathrm{LP}$}
\end{array}
\end{equation}
For a ghost value vector $g\in \mathbb{Z}_{\geq 0}^E$, we call the constraint $x(\delta(S)) \geq k - g(\delta(S))$ the \emph{$g$-cut constraint for $S$}. We say that such a constraint is \emph{$y$-tight} if $y(\delta(S)) = k - g(\delta(S))$. Analogously, we will say that a constraint $x_e = p_e$ for $p_e \in \mathbb{Z}_{\geq 1}$ on a single edge $e$ is \emph{$y$-tight} if $y_e = p_e$.


Our algorithm is formally described in \Cref{alg:main} and uses the following notation. For any vector $y\in \mathbb{R}^E$, we denote by $\fracpart(y)\coloneqq \{e\in E\colon y_e \not\in \mathbb{Z}\}$ all edges with a fractional $y$-value. To clearly distinguish between the original input graph and the graph at each iteration of our algorithm, we will denote the original input graph by $\overline{G}=(\overline{V},\overline{E})$  and the graph used in each iteration of the algorithm (i.e., $\overline{G}$ after some vertex contractions and edge deletions) by $G=(V,E)$. We will be explicit about the specific graph considered. For any two vertex sets $S,T \subseteq V$, we denote by $E(S,T)\subseteq E$ all edges with one endpoint in $S$ and one in $T$.
For vertices $u,v$ and $S\subseteq V$, we also use the shorthand $E(u,v) \coloneqq E(\{u\}, \{v\})$ and $E(u,S) \coloneqq E(\{u\},S)$. For $S \subseteq V$, we let $E[S] := \{\{u, v\} \in E : u,v \in S\}$ be all edges with both endpoints in $S$. We note that even though one could perform multiple ghost value augmentations or multiple cut relaxations in a single iteration of the while loop, \cref{algline:ghostValueAugm} only performs a single such operation per iteration for simplicity.

\DontPrintSemicolon%
\begin{algorithm2e}
  \renewcommand{\Return}{\textbf{Return}\xspace}
  \SetKwBlock{Initialization}{Initialization}{}

  \Initialization{
    $z(e)=0$ for all $e\in \overline{E}$. \tcp*[f]{$z\in \mathbb{Z}_{\geq 0}^{\overline{E}}$ will be the returned solution.}\;
    $g(e)=0$ for all $e\in \overline{E}$. \tcp*[f]{Start with all ghost values being set to $0$.}\;
    $G = (V,E) \coloneqq (\overline{V}, \overline{E})$.\;
    Let \LPA be~\ref{lp:ghost} using $y$ and update $y$ to an optimal vertex solution.\label{algline:solveLPInit}\;
    Delete from $G$ all edges $e\in E$ with $y_e=0$.\;
  }
  \smallskip

  \While{$y\not\in \mathbb{Z}^E$}{
    \uIf{there are distinct vertices $u,v\in V$ with  $(y+g)(E(u,v)) \in \left[\frac{k}{2} - 2, \frac{k}{2}\right)$}{
      \textbf{Ghost Value Augmentation:} set $g_e = g_e + 2$ for an arbitrary edge $e\in E(u,v)$.
    }\label{algline:ghostValueAugm}
  \ElseIf{there is a $y$-tight $g$-cut constraint $x(\delta(S))\geq k - g(\delta(S))$ in \LPA such that\label{algline:ifSetRelax}
  \begin{enumerate}[label= (\roman*), nosep]
    \item \LPA does not contain a $y$-tight $g$-cut constraint $x(\delta(T)) \geq k - g(\delta(T))$ for $T\subsetneq S$, and
  \item $|\fracpart(y)\cap \delta(S)| \leq 3$,
\end{enumerate}}{
  $z_e = y_e$ for all $e\in E[S]$.\;
  \textbf{Drop/Contract:} Contract set $S$ in $G$ and remove $g$-cut constraint for $S$ from \LPA.\;
}\label{algline:relaxation}

  Compute an optimal vertex solution $y$ to \LPA (with ghost values g).\;
  Delete from $G$ (and from \LPA) all edges $e\in E$ with $y_e + g_e=0$.\label{algline:resolveLP}\;

   For all $e\in E$ with $y_e\in \mathbb{Z}$, we add the constraint $x_e=y_e$ to \LPA.\label{algline:addEqualityConstr}
  }
  \smallskip

  $z_e = y_e$ for all $e\in E$.\;

  \smallskip

  \Return~$z$.\;

  \caption{Main Algorithm. ($k$ is even.)}\label{alg:main}
\end{algorithm2e}

\paragraph*{Algorithm Intuition.} We summarize the intuition for our algorithm. As discussed earlier, a natural approach to rounding our vector $y$ would be to argue that whenever we re-solve our LP and it does not have a newly integral edge, there must be a constraint of our LP corresponding to a set that has at most $O(1)$ fractional edges crossing it. Such a constraint can be safely dropped from our LP and doing so corresponds to our drop/contract relaxation. However, standard arguments that the non-existence of such a set implies a newly integral edge would require showing that the relevant cut constraints can be uncrossed into a laminar family. This uncrossing is not necessarily possible if we have already dropped some of our cut constraints.  The somewhat unusual operation of ghost value augmentations will rescue us from this situation. In particular, we will see that the cases where uncrossing fails are exactly those when we are able to perform a ghost value augmentation. Indeed, a ghost value augmentation can be thought of as its own sort of cut relaxation: putting a ghost value of $2$ units on some edge $e\in E$, simply corresponds to replacing the cut constraints $x(\delta(S)) \geq k$ by $x(\delta(S)) \geq k-2$ for all $S\subseteq V$ with $e\in \delta(S)$.

\section{Analysis of Our Algorithm}\label{sec:proof}
We proceed to analyze our algorithm. Doing so will require addressing several non-trivial issues. First, it is not a priori clear that our algorithm terminates and, in particular, that we can always make progress by some relaxation or edge deletion or freezing. Second, even if the algorithm terminates, it is not clear that the returned vector $z$ is integral, much less that it has value at least $k-O(1)$ on every cut. For integrality, we will need to argue that, whenever we contract a set, all of its internal edges are integral. For near-$k$-edge-connectivity, it is particularly not clear that we do not end up with a single cut across which we have performed many ghost value augmentations, and so we will need to argue that this does not happen. Lastly, there are several efficiency concerns to address, including how to find the set $S\subseteq V$ of \cref{algline:relaxation} if there is one satisfying the stated criteria.

Before addressing these challenges, we introduce some notation we use throughout our analysis. 
For brevity, when saying that a property holds at \emph{any iteration of the algorithm}, we mean that it holds at the beginning of any iteration of the while loop in \cref{alg:main}.
As in \cref{alg:main}, we will let \LPA be the LP used by our algorithm at the beginning of an iteration. Note that although we compute our vertex solution $y$ at the end of an iteration and then possibly delete edges from $G$ and \LPA, even after deleting said edges $y$ remains a vertex to the resulting \LPA.\footnote{This can be seen by, e.g., observing that $y$ is a vertex solution if and only if it is feasible and it is the unique solution to $\bar{A}x = b$ (i.e., the columns of $\bar{A}$ are independent) where $\bar{A}$ is the subsystem of constraints of $A$ tight for $y$. For \LPA, the columns of $\bar{A}$ remain independent even after deleting the column and row corresponding to an $e$ with $y_e = 0$ so $y$ without this coordinate remains a vertex solution.} In other words, $y$ is a vertex solution to \LPA at the beginning of each iteration.
Likewise, we denote by $G=(V,E)$, $y\in \mathbb{R}_{\geq 0}^E$, and $g\in \mathbb{Z}_{\geq 0}^E$ the current graph, the current vertex solution to \LPA, and the current ghost values, respectively, at that iteration. Throughout our analysis, we will represent by $\mathcal{R} \subseteq 2^{\overline{V}\setminus \{r\}}$ all sets that are contracted throughout the algorithm.
More precisely, for some set $R\subseteq \overline{V}\setminus \{r\}$, we have $R\in \mathcal{R}$ if $R$ is not a singleton and during some point in the algorithm we had a vertex that, when undoing the contractions, corresponds to the vertex set $R$.
Because we perform contractions consecutively, the family $\mathcal{R}$ is laminar.
At any iteration of the algorithm, each vertex $v\in V$ of the current graph $G=(V,E)$ is either an original vertex, i.e., $v\in \overline{V}$, or it corresponds to a set $R\in \mathcal{R}$ that was contracted in a prior iteration of the algorithm.

\subsection{Termination of Algorithm}\label{sec:termination}
We start by showing that \cref{alg:main} terminates.

To show that \cref{alg:main} terminates, we show that, at any iteration of the while loop, if $y$ is not yet integral, and we cannot perform a ghost value augmentation, then we can perform a cut relaxation.
Because cut relaxations require cut constraints with high integrality, i.e., the number of fractional edges must be at most $3$, we will derive sparsity results in the following.
These results provide upper bounds on the number of $y$-fractional edges, or show that certain edges must have integral $y$-values.

We start with a basic property on the $(y+g)$-load on each cut. A consequence is that at any iteration, we have $(y+g)(\delta(v)) \ge k-2$ for any vertex $v$, and for every $S \subseteq V$ with $2 \le |S| \le n-2$ we have $(y+g)(\delta(v)) \ge k$. Thus our graph is close to being fractionally $k$-edge-connected.
\begin{lemma}\label{lem:yPlusGHigh}
  At any iteration of \cref{alg:main}, we have
  \begin{equation*}
    (y+g)(\delta(S)) \geq k -2 \qquad \forall S\subseteq V\setminus \{r\}, S\neq\emptyset.
  \end{equation*}
  Also, if for any non-empty set $S\subseteq V\setminus \{r\}$, we have $(y+g)(\delta(S)) = k-2$, then $y_e\in \mathbb{Z}_{\geq 0} \;\forall e\in \delta(S)$.
\end{lemma}
\begin{proof}
  If the $g$-cut constraint corresponding to $S$ is still in \LPA, then we even have $(y+g)(\delta(S)) \geq k$.
  Otherwise, let $\widetilde{y}$ and $\widetilde{g}$ be the \LPA solution and ghost values, respectively, at the iteration when the $g$-cut constraint corresponding to $S$ got relaxed.
  Hence, $(\widetilde{y}+\widetilde{g})(\delta(S)) \geq k$, because, as before, the $g$-cut constraint corresponding to $S$ is part of \LPA at the beginning of the iteration when $S$ gets relaxed.
  Moreover, $|\fracpart(\widetilde{y})\cap \delta(S)| \leq 3$, which implies
  \begin{equation*}
    (\lfloor \widetilde{y} \rfloor + \widetilde{g})(\delta(S)) > k-3.
  \end{equation*}
  Because the left-hand side is integral, we get
  \begin{equation*}
    (\lfloor \widetilde{y} \rfloor + \widetilde{g})(\delta(S)) \geq k-2.
  \end{equation*}
  The first statement now follows by observing that integral $y$-values get fixed, and hence $y \geq \lfloor \widetilde{y} \rfloor$, and that $g \geq \widetilde{g}$, because ghost values are non-decreasing.
  Moreover, this reasoning shows that to get $(y+g)(\delta(S)) = k-2$, we need $y(\delta(S)) = \lfloor \widetilde{y} \rfloor (\delta(S))$.
  Because $y\geq \lfloor \widetilde{y} \rfloor$, this implies as desired $y_e = \lfloor \widetilde{y}_e \rfloor \in \mathbb{Z}_{\geq 0} \;\forall e\in \delta(S)$.
\end{proof}

The following lemma formalizes that $y$ has low fractionality on any set of parallel edges. Below, recall that $y$ is a vertex solution to the relevant LP in each iteration.
\begin{lemma}\label{lem:atMostOneFracInEuv}
  At any iteration of the algorithm, we have
  \begin{equation*}
    |\fracpart(y) \cap E(u,v)| \leq 1 \qquad \forall u,v\in V, u\neq v.
  \end{equation*}
\end{lemma}
\begin{proof}
  Assume for the sake of deriving a contradiction that there is a pair of vertices $u,v \in V, u\neq v$ with $|\fracpart(y)\cap E(u,v)| \geq 2$.
  Let $e_1, e_2 \in \fracpart(y)\cap E(u,v)$ be two distinct edges.
  Observe that for
  \begin{equation*}
    \epsilon = \min\left\{y_{e_1} - \lfloor y_{e_1}\rfloor, \lceil y_{e_1} \rceil - y_{e_1}, y_{e_2} - \lfloor y_{e_2} \rfloor, \lceil y_{e_2}\rceil - y_{e_2}\right\} > 0,
  \end{equation*}
  we have that both $y + \epsilon (\chi^{\{e_1\}} - \chi^{\{e_2\}})$ and $y - \epsilon (\chi^{\{e_1\}} - \chi^{\{e_2\}})$ are feasible for the LP.
  This holds because there is no constraint that contains $e_1$ and not $e_2$ or vice-versa, except for integral bounds (lower/upper bounds, or equality constraints) on the values of $y_{e_1}$ and $y_{e_2}$.
  This contradicts that $y$ is a vertex solution to the LP.
\end{proof}

To obtain strong enough sparsity to show that a cut relaxation is possible, we use a classic strategy.
Namely, we first show that, at any iteration of \cref{alg:main}, a vertex solution can be described as the unique solution to a very structured system of $y$-tight LP constraints.
More precisely, the $y$-tight $g$-cut constraints in this system correspond to sets that form a laminar family.
This is where we crucially exploit the use of ghost values, without which a statement as below would be wrong, as demonstrated in \cref{fig:lowfracissue}. 

\begin{figure}[htb!]
\centering 
  \begin{tikzpicture}[node distance=2.5cm, line width=1pt]
  \begin{scope}[every node/.style={inner sep=3mm}]
   \node (a) {};
   \node (b) [right of=a] {};
   \node (c) [right of=b] {};
  \end{scope}
   \node (invis-right) [right of=c,draw=none] {};
   \node (invis-right-up) [right of=c,yshift=10mm,xshift=-3mm,draw=none] {};
   \node (invis-left) [left of=a,draw=none] {};
   \node (invis-up) [above of=b,yshift=-5mm,draw=none] {};

   \begin{scope}[-,blue]
     \draw (invis-left) -- (a) node[midway,above,xshift=-0.1cm] {$\frac{k}{2}-1$};
     \draw (a) -- (b) node[midway,above,xshift=-5.0mm,yshift=-0.7mm] {$\frac{k}{2}-1$};
     \draw (b) -- (c) node[midway,above,xshift=0.3cm] {$\frac{k}{2}-1$};
     \draw (c) -- (invis-right) node[midway,above,xshift=0.3cm] {$\frac{k}{2}-1$};
   \end{scope}

   \begin{scope}[-,dashed,bend right=15]
     \draw (invis-left) to node[below] {} (a);
     \draw (b) to node[below right] {} (c);
     \draw (c) to node[below,xshift=6mm] {} (invis-right);
     \draw (b) to node[right] {} (invis-up);
   \end{scope}
   \begin{scope}[-,dashed,bend right=30]
		\draw (b) to node[below right] {} (c);
		\draw (c) to node[below,xshift=6mm] {} (invis-right);
   \end{scope}
   \draw[-,dashed] (c) to[bend left=15] node[below,xshift=6mm] {} (invis-right-up);

  \node[ellipse, thick, red, draw, fit=(a) (b), inner sep=0pt, minimum height=16mm, label={[xshift=-1.0cm, yshift=-0.1cm, red]$S$}] {};
  \node[ellipse, thick, green!50!black, draw, fit=(b) (c), inner sep=0pt, minimum height=16mm,label={[xshift=1.0cm, yshift=-0.1cm, green!50!black]$T$}] {};
\end{tikzpicture}
\caption{A situation in which ghost value augmentation is necessary. All dotted edges have value 1/2, so $S$ and $T$ are tight but not dropped. $S \cap T$ and $S \smallsetminus T$, however, have been dropped and thus are allowed to drop below connectivity $k$. Therefore it is not possible to uncross $S$ and $T$. Note that the integrality of the rightmost blue edge leaving $T$ is not necessary. In particular, this rightmost blue edge could be any number of fractional edges and the instance would have the same behavior.}\label{fig:lowfracissue}
\end{figure}
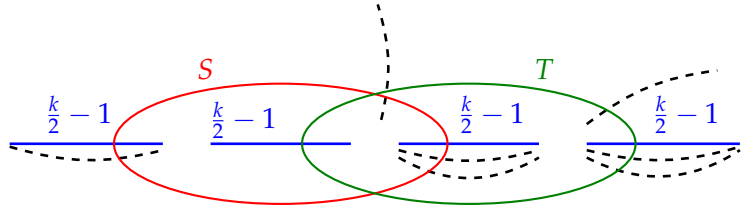

\begin{lemma}\label{lem:lamTightSys}
  Consider any iteration of \cref{alg:main} where no ghost value augmentation can be performed.
  Let $\mathcal{L}\subseteq \{S \subseteq V\setminus \{r\} \colon y(\delta(S)) = k - g(\delta(S))\}$ be a maximal laminar family corresponding to $y$-tight $g$-cut constraints in \LPA.
  Then, the linear equation system consisting of 
  \begin{equation*}
    x(\delta(L)) = k - g(\delta(L)) \quad \forall L\in \mathcal{L},
  \end{equation*}
  together with all $y$-tight constraints on single edges in \LPA, i.e., constraints of type $x_e = p_e$ for some $e\in E$ and $p_e\in \mathbb{Z}$, form a full column rank system\footnote{A system of linear equations $Ax=b$ has full column rank if the rank of $A$ is equal to the number of columns of $A$.} of $y$-tight constraints of \LPA.
  Thus, because $y$ is a vertex solution to \LPA, it is the unique solution to this system.
\end{lemma}
\begin{proof}
  Assume for the sake of contradiction that there is an iteration of the algorithm where no ghost value augmentation can be performed, and such that there is a maximal laminar family $\mathcal{L} \subseteq 2^{V\setminus \{r\}}$ corresponding to $y$-tight $g$-cut constraints in \LPA fulfilling the following:
  The linear equation system consisting of all equations corresponding to $y$-tight $g$-cut constraints of sets in $\mathcal{L}$, together with equations corresponding to all $y$-tight constraints on single edges in \LPA, does not have full column rank.
  We call this linear equation system the \emph{reference system}.

  The reference system not having full column rank implies that there must be a $y$-tight $g$-cut constraint $x(\delta(S)) \geq k - g(\delta(S))$ not implied by it.
  Among all such sets $S\subseteq V\setminus \{r\}$, we choose one where
  \begin{equation*}
    \mathcal{L}_S \coloneqq \{L \in \mathcal{L} \colon \text{$L$ and $S$ are crossing}\}
  \end{equation*}
  has smallest cardinality.
  Because the constraint $x(\delta(S)) \geq k - g(\delta(S))$ is not implied by the reference system, and we did not add that constraint to it, it must be that $S$ crosses some set in $\mathcal{L}$.
  Hence, $|\mathcal{L}_S| \geq 1$.
  Let $L\in \mathcal{L}_S$.
  We continue by showing that the $g$-cut constraints corresponding to certain sets must have been dropped, allowing us to reduce to a setting similar to \cref{fig:lowfracissue}.

  \begin{claim}\label{claim:singletonSets} We have
    \hspace{0mm}\begin{enumerate}[label= (\roman*),nosep]
      \item\label{item:SintLSingleton}
        \begin{itemize}[nosep]
          \item The $g$-cut constraint corresponding to $S\cap L$ has been dropped from \LPA and the equation $x(\delta(S\cap L)) = k - g(\delta(S\cap L))$ is not implied by the reference system, 
          \item $(y+g)(\delta(S\cap L)) \leq k$, and
          \item if $(y+g)(\delta(S\cap L)) = k$, then $(y+g)(\delta(S\cup L))=k$ and $E(S\setminus L, L\setminus S) = \emptyset$.
        \end{itemize}
      \item\label{item:oneOfSymDiffSingleton} There is a set $Q_1$ among $\{S\setminus L, L\setminus S\}$ (let $Q_2$ be the other set) such that
        \begin{itemize}[nosep]
          \item the $g$-cut constraint corresponding to $Q_1$ has been dropped from \LPA,
          \item $(y+g)(\delta(Q_1)) \leq k$, and
          \item if $(y+g)(\delta(Q_1)) = k$, then also $(y+g)(\delta(Q_2))=k$ and $E(S\cap L, V\setminus (S\cup L)) = \emptyset$.
        \end{itemize}
    \end{enumerate}
  \end{claim}
  \begin{proof}
    We start by proving~\ref{item:SintLSingleton}.
    If $(y+g)(\delta(S\cap L)) < k$, then the $g$-constraint corresponding to $S\cap L$ must have been dropped, and all conditions of~\ref{item:SintLSingleton} are fulfilled.
    Hence, assume from now on $(y+g)(\delta(S\cap L)) \geq k$.
    
    We use the following well-known basic relation, which can be verified by checking that the shown equation holds for each coordinate (note that every coordinate corresponds to an edge): 
    \begin{equation}\label{eq:basicCupCapRel}
      \chi^{\delta(S)} + \chi^{\delta(L)} = \chi^{\delta(S\cup L)} + \chi^{\delta(S\cap L)} + 2 \chi^{E(S\setminus L, L \setminus S)}.
    \end{equation}
    By taking the scalar product of the above equation with $y + g$, we get
    \begin{equation}\label{eq:ygCapCup}
      (y+g)(\delta(S)) + (y+g)(\delta(L)) = (y+g)(\delta(S\cup L)) + (y+g)(\delta(S\cap L)) + 2(y+g)(E(S\setminus L, L\setminus S)).
    \end{equation}
    Note that the $g$-cut constraint corresponding to $S\cup L$ has not been dropped yet because $|S\cup L| \geq 2$.
    Hence, $(x+g)(\delta(S\cup L))\geq k$.
    Together with $(y+g)(\delta(S)) = (y+g)(\delta(L)) = k$ and $(y+g)(\delta(S\cap L)) \geq k$, \Cref{eq:ygCapCup} thus implies $(y+g)(\delta(S\cap L)) = (y+g)(\delta(S\cup L)) = k$ and $(y+g)(E(S\setminus L, S\setminus L))=0$.
    The last equation implies $E(S\setminus L, L\setminus S) = \emptyset$, because $\supp(y+g)=E$, which holds because we deleted all edges with $(y+g)$-value zero.
    Hence,~\eqref{eq:ygCapCup} simplifies to
    \begin{equation*}
      \chi^{\delta(S)} + \chi^{\delta(L)} = \chi^{\delta(S\cup L)} + \chi^{\delta(S\cap L)}. 
    \end{equation*}
    Because $\chi^{\delta(L)}$ is a row of our reference system and $\chi^{\delta(S)}$ is not spanned by the rows of our reference system, we have that either the equation $x(\delta(S\cup L)) = k - g(\delta(S\cup L))$ or $x(\delta(S\cap L))=k - g(\delta(S\cap L))$ is not implied by our reference system.
    Note that both equations correspond to $y$-tight $g$-cut constraints as shown above.
    Because the $g$-cut constraint corresponding to $S\cup L$ is still part of \LPA as $|S\cup L|\geq 2$, and $\mathcal{L}_{S\cup L}\subsetneq \mathcal{L}_S$,\footnote{%
      $\mathcal{L}_{S\cup L}\subsetneq \mathcal{L}_S$ (and analogously $\mathcal{L}_{S\cap L}, \mathcal{L}_{S\setminus L}, \mathcal{L}_{L\setminus S} \subsetneq \mathcal{L}_S$) follows from the following observation on laminar families.
    Let $\mathcal{L}$ be a laminar family over some finite ground set $N$, and let $S\subseteq N$ and $L\in \mathcal{L}$ be such that $S$ crosses $L$.
    Then any set in $\mathcal{L}$ that crosses $S\cup L$ also crosses $S$.
    However, the set $L$, which crosses $S$, does not cross $S\cup L$.
    } %
    this $g$-cut constraint must be implied by the reference system.
    For otherwise, we could have chosen $S\cup L$ instead of $S$, which violates our choice of $S$.
    Hence, the $g$-cut constraint corresponding to $S\cap L$ is not implied by our reference system.
    Because also $\mathcal{L}_{S\cap L} \subsetneq \mathcal{L}_S$, the $g$-cut constraint corresponding to $S\cap L$ cannot be part of \LPA anymore, as this would again violate our choice of $S$.

    \medskip

    For point~\ref{item:oneOfSymDiffSingleton} we can follow an analogous approach as for~\ref{item:SintLSingleton}.
    If there is a set $Q\in \{S\setminus L, L\setminus S\}$ with $(y+g)(\delta(Q)) <k$, then we choose $Q_1\coloneqq Q$.
    Indeed, the $g$-cut constraint corresponding to $Q_1$ must have been dropped from \LPA because $(y+g)(\delta(Q_1))<k$.
    Thus, $Q_1$ fulfills all conditions of point~\ref{item:oneOfSymDiffSingleton}.
    Hence, from now on assume $(y+g)(\delta(S\setminus L)) \geq k$ and $(y+g)(\delta(L\setminus S)) \geq k$.

    We use the following well-known basic relation:
    \begin{equation*}
      \chi^{\delta(S)} + \chi^{\delta(L)} = \chi^{\delta(S\setminus L)} + \chi^{\delta(L\setminus S)} + 2 \chi^{E(S\cap L, V\setminus (S\cup L))}.
    \end{equation*}
    Because $(y+g)(\delta(S))=k$ and $(y+g)(\delta(L))=k$, we must have $(y+g)(\delta(S\setminus L)) = k$, $(y+g)(\delta(L\setminus S)) = k$, and $E(S\cap L, V\setminus (S\cup L)) = \emptyset$.
    This implies
    \begin{equation*}
      \chi^{\delta(S)} + \chi^{\delta(L)} = \chi^{\delta(S\setminus L)} + \chi^{\delta(L\setminus S)}.
    \end{equation*}
    Because $\chi^{\delta(L)}$ is a row in our reference system and $\chi^{\delta(S)}$ is not spanned by the rows of your reference system, we have that either the equation $x(\delta(S\setminus L)) = k - g(\delta(S\setminus L))$ or $x(\delta(L\setminus S)) = k - g(\delta(L\setminus S))$ (or both) is not implied by the reference system.
    Note that both equations correspond to $y$-tight $g$-cut constraints as shown above.
    Let $Q_1\in \{S\setminus L, L\setminus S\}$ be such that $x(\delta(Q_1)) = k - g(\delta(Q_1))$ is not implied by the reference system.
    Because $\mathcal{L}_{Q_1}\subsetneq \mathcal{L}_S$, the $g$-cut constraint corresponding to $Q_1$ must have been dropped from \LPA.
    For otherwise, we could have chosen $Q_1$ instead of $S$, which violates our choice of $S$.
    Hence, $Q_1$ fulfills all conditions of point~\ref{item:oneOfSymDiffSingleton}.\qedhere~(\cref{claim:singletonSets})
  \end{proof}

By \cref{claim:singletonSets}~\ref{item:SintLSingleton}, the $g$-cut constraint corresponding to $S\cap L$ got relaxed/dropped from \LPA and $x(\delta(S \cap L)) = k - g(\delta(S\cap L))$ is not implied by the reference system.
Moreover, \cref{claim:singletonSets}~\ref{item:oneOfSymDiffSingleton} implies that the constraint corresponding to at least one of the sets $S\setminus L$ or $L\setminus S$ got dropped.
Recall that dropped sets got contracted and therefore correspond to singletons.
We finish the proof by showing that a ghost value augmentation could have been performed with respect to the singleton $S\cap L$ and either $S\setminus L$ or $L\setminus S$.

  To this end consider the different $(y+g)$-loads between the four sets $S\setminus L$, $S\cap L$,  $L\setminus S$, and $V\setminus (S\cup L)$.
  For brevity let

  \noindent\begin{minipage}{0.5\linewidth}
    \begin{align*}
      a &\coloneqq (y+g)(E(S\cap L, S\setminus L))\\
      b &\coloneqq (y+g)(E(S\cap L, L\setminus S))\\
      c &\coloneqq (y+g)(E(L\setminus S, V\setminus (S\cup L)))\\
      d &\coloneqq (y+g)(E(S\setminus L, V\setminus (S\cup L)))\\
      e &\coloneqq (y+g)(E(S\setminus L, L\setminus S))\\
      f &\coloneqq (y+g)(E(S\cap L, V\setminus (S\cup L))).
    \end{align*}
  \end{minipage}%
  \begin{minipage}{0.5\linewidth}
    \vspace{5mm}
    \begin{tikzpicture}

  \begin{scope}[line width=1pt]

    \begin{scope}
      \coordinate (c1) at (0,0);
      \coordinate (c2) at (2,0);

      \draw (c1) ellipse (1.5cm and 1.0cm);
      \draw (c2) ellipse (1.5cm and 1.0cm);
    \end{scope}

    \begin{scope}
      \node at ($(c1) + (-1.2cm, -1.0cm)$) {$S$};
      \node at ($(c2) + (1.2cm, -1.0cm)$) {$L$};
    \end{scope}

    \coordinate (c) at ($(c1)!0.5!(c2)$);
    \coordinate (out) at ($(c) + (0cm, 1.5cm)$);

    \begin{scope}[red]
      \def\s{2mm}
      \draw ($(c1)+(-2*\s,0)$) to[bend right] node[below] {$a$} ($(c)+(-\s,0)$);
      \draw ($(c2)+(2*\s,0)$) to[bend left] node[below] {$b$} ($(c)+(\s,0)$);
      \draw ($(c2)+(2*\s,2*\s)$) to[bend right] node[pos=0.8,right] {$c$} (out-|c2);
      \draw ($(c1)+(-2*\s,2*\s)$) to[bend left] node[pos=0.8,left] {$d$} (out-|c1);
      \draw ($(c1)+(-\s,-3*\s)$) to[out=-35, in=-135] node[pos=0.5,below] {$e$} ($(c2)+(\s,-3*\s)$);
      \draw (out) to[bend left] node[pos=0.2,right] {$f$} ($(c)+(0,\s)$);
    \end{scope}

  \end{scope}

\end{tikzpicture}
  \end{minipage}

\bigskip

  The above definitions immediately lead to the following basic relations
  \begin{equation}\label{eq:abcdefRel}
    \begin{array}{l@{\;}c@{\;}l@{\;}c@{\;}l}
      k     &=    &(y+g)(\delta(S))            &= &b + d + e + f \\
      k     &=    &(y+g)(\delta(L))            &= &a + c + e + f \\
      k     &\geq &(y+g)(\delta(S\cap L))      &= &a + b + f     \\
      k - 2 &\leq &(y+g)(\delta(S\cap L))      &= &a + b + f     \\
      k - 2 &\leq &(y+g)(\delta(S\setminus L)) &= &a + d + e     \\
      k - 2 &\leq &(y+g)(\delta(L\setminus S)) &= &b + c + e, 
    \end{array}
  \end{equation}
  where the equalities at the start of the first two lines hold because $S$ and $L$ correspond to $y$-tight $g$-cut constraints,
  the inequality in the third line follows from \cref{claim:singletonSets}~\ref{item:SintLSingleton},
  and the inequalities at the start of the last three lines hold by \cref{lem:yPlusGHigh}.

  By subtracting the first relation above from the sum of the forth and fifth one, we get
  \begin{equation}\label{eq:aLarge}
    a \geq \frac{k}{2} - 2.
  \end{equation}

  Analogously, by subtracting the second relation in~\eqref{eq:abcdefRel} from the sum of the forth and sixth one, we get
  \begin{equation}\label{eq:bLarge}
    b \geq \frac{k}{2} - 2.
  \end{equation}

  Let $Q_1\in \{S\setminus L, L \setminus S\}$ be the set as described in \cref{claim:singletonSets}~\ref{item:oneOfSymDiffSingleton}.
  We will show that a ghost value augmentation could have been performed between $Q_1$ and $S\cap L$.
  Note that these two sets correspond to singletons, and~\eqref{eq:aLarge}/\eqref{eq:bLarge} show that the (parallel) edges between them fulfill the lower bound condition necessary to apply a ghost value augmentation.
  It remains to show that, if $Q_1 = S\setminus L$, we have $a < \frac{k}{2}$, and analogously, if $Q_1 = L\setminus S$, then $b< \frac{k}{2}$.
  Then all conditions are fulfilled to apply a ghost value augmentation, which is the desired contradiction.

  The proof for the two cases is identical; hence, we assume from now on $Q_1 = S\setminus L$.
  We have
  \begin{align*}
    a = (y+g)(E(S\setminus L, S\cap L))
      = \frac{1}{2}\left[ (y+g)(\delta(S\cap L)) + (y+g)(\delta(S\setminus L)) - (y+g)(\delta(S))\right]
      \leq \frac{k}{2},
  \end{align*}
  where the inequality follows from the first and third relation in~\eqref{eq:abcdefRel}, and from $(y+g)(\delta(S\setminus L)) \leq k$, which holds by \cref{claim:singletonSets}~\ref{item:oneOfSymDiffSingleton} and the assumption $Q_1 = S\setminus L$.
  It remains to show that $a\neq \frac{k}{2}$.
  Assume for the sake of deriving a contradiction that $a = \frac{k}{2}$.
  We therefore must have $(y+g)(\delta(S\cap L)) = k$ and $(y+g)(\delta(S\setminus L)) =k$, which implies by \cref{claim:singletonSets}~\ref{item:oneOfSymDiffSingleton} in particular $(y+g)(\delta(S\cup L))=k$, $(y+g)(\delta(L\setminus S))=k$, and $E(S\cap L, V\setminus (S\cup L)) = \emptyset$.
  The contradiction we derive will be that the equation $x(\delta(S\cap L)) = k - g(\delta(S\cap L))$ is implied by \LPA, which violates \cref{claim:singletonSets}~\ref{item:SintLSingleton}.
  Note that because $E(S\cap L, V\setminus (S\cup L)) = \emptyset$, we get
  \begin{equation}\label{eq:descSIntL}
    \chi^{\delta(S\cap L)} = \chi^{E(S\setminus L, S\cap L)} + \chi^{E(L\setminus S, S\cap L)}.
  \end{equation}
  Moreover, $a=\frac{k}{2}\in \mathbb{Z}$ (recall we are assuming $k$ is even) implies $y(e)\in \mathbb{Z}$ for $e\in E(S\setminus L, S\cap L)$ because of \cref{lem:atMostOneFracInEuv}.
  Hence, the $y$-values on the edge $E(S\setminus L, S\cap L)$ are fixed by the integrality constraints, which are part of the reference system.
  Moreover, we also have $b= (x+y)(\delta(S\cap L)) - a = \frac{k}{2}$, where the first equality holds because $E(S\cap L, V\setminus (S\cup L))=\emptyset$.
  Hence, if also the $g$-cut constraint corresponding to $L\setminus S$ got dropped from \LPA, then we have analogously $y(e)\in \mathbb{Z}$ for $e\in E(L\setminus S, S\cap L)$.
  However, then~\eqref{eq:descSIntL} implies that $x(\delta(S\cap L)) = k - g(\delta(S\cap L))$ is implied by \LPA, which is a contradiction.
  Thus, it remains to consider the case where the $g$-cut constraint corresponding to $L\setminus S$ is still part of \LPA.
  Note that
  \begin{equation*}
    \chi^{E(L\setminus S, S\cap L)} = \frac{1}{2} \left[ \chi^{\delta(L\setminus S)} + \chi^{\delta(S)} - \chi^{\delta(S\cup L)} \right].
  \end{equation*}
  Observe that $\chi^{\delta(L\setminus S)}$, $\chi^{\delta(S\cup L)}$, and $\chi^{\delta(S)}$ are all row vectors of our reference system.
  Hence, also the row vector $\chi^{E(L\setminus S, S\cap L)}$ is implied by our reference system.
  This in turn implies by~\eqref{eq:descSIntL} that the equation $x(\delta(S\cap L)) = k - g(\delta(S\cap L))$ is implied by our reference system, which leads to the desired contradiction.
  \hfill\qedhere~(\cref{lem:lamTightSys})
\end{proof}

For completeness, we now present a classic reasoning, adjusted to our context, showing that the sparsity of \LPA at any iteration of the algorithm can be bounded by the number of linearly independent $y$-tight $g$-cut constraints in a full column rank equation system of tight LP constraints.
We state the result for an equation system defining the \LPA vertex $y$ with an arbitrary family $\mathcal{L}$ corresponding to $y$-tight $g$-cut constraints.
We later apply the statement with an equation system of $y$-tight constraints coming from \cref{lem:lamTightSys}, where the family $\mathcal{L}$ is laminar.
\begin{lemma}\label{lem:sparsityFromLamFamily}
  Consider an iteration of \cref{alg:main} where $y$ is not yet integral, and no ghost value augmentation can be performed.
  Consider a full column rank system of equations corresponding to $y$-tight constraints of \LPA, and let $\mathcal{L}\subseteq 2^{V\setminus \{r\}}$ be all cuts of $y$-tight $g$-cut constraints that correspond to an equation in the equation system.
  Then $|\fracpart(y)|\leq |\mathcal{L}|$.
\end{lemma}
\begin{proof}
  First, we can assume that the considered equation system, for simplicity we call it the \emph{reference system}, is a square system.
  Indeed, if the reference system is not square, then we can successively remove equations from the system that are implied by the other equations of the system until we get a square system.
Moreover, the implication of the statement for the square system implies the statement for the original system.

  The square reference system has two types of equations:
  \begin{itemize}[nosep,rightmargin=2mm]
    \item equations corresponding to $y$-tight $g$-cut constraints, i.e., $x(\delta(S)) = k - g(\delta(S)) \;\forall S\in \mathcal{L}$, and
    \item equations corresponding to $y$-tight constraints on single edges, i.e., these are of the form $x_e = p_e$, where $e\in E$ and $p_e\in \mathbb{Z}_{\geq 1}$.
  \end{itemize}
  Let $F\subseteq E$ be all edges for which an equation of the second type is in the system.
  Because the reference system is square, we have $|E| = |\mathcal{L}| + |F|$.
  Moreover, all edges in $F$ have integral $y$-values, which implies the result because of
  \begin{equation*}
    |\fracpart(y)| \leq |E| - |F| = |\mathcal{L}|. \qedhere
  \end{equation*} 
\end{proof}

Finally, the following lemma shows that we make progress in each iteration of \cref{alg:main}.
\begin{lemma}\label{lem:progress}
  At any iteration of \cref{alg:main}, if $y$ is not integral and no ghost value augmentation can be applied (i.e., the algorithm is at an iteration where it reaches \cref{algline:ifSetRelax}), then there is a $g$-cut constraint of \LPA that can be relaxed.
\end{lemma}
\begin{proof}
  Let $\mathcal{L} \subseteq \{S\subseteq V\setminus \{r\}\colon x(\delta(S)) = k - g(\delta(S))\}$ be a maximal laminar family of sets corresponding to $y$-tight $g$-cut constraints of \LPA.
  By \cref{lem:lamTightSys}, these constraints, together with all $y$-tight constraints of \LPA on single edges, correspond to a full column rank equation system with $y$ being its unique solution.
  We successively remove from $\mathcal{L}$ constraints that are redundant in that system, until no $g$-cut constraint corresponding to a set in $\mathcal{L}$ is redundant.

  Assume for the sake of deriving a contradiction that no cut relaxation can be applied.
  We will derive a contradiction by showing that this would imply $|\mathcal{L}| > |\fracpart(y)|$, which contradicts \cref{lem:sparsityFromLamFamily}.
  To this end we use a token counting argument, where we assign two tokens to each edge in $\fracpart(y)$, and assign those tokens to the sets $\mathcal{L}$ such that each set in $\mathcal{L}$ gets at least $2$ tokens, and at least one set gets strictly more than $2$ tokens.

  First, each edge $\{u,v\} \in \fracpart(y)$ assigns one token to the smallest set in $\mathcal{L}$ that contains $u$ (if there is such a set) and one to the smallest set in $\mathcal{L}$ that contains $v$ (if there is such a set).
  We then consider the sets in $\mathcal{L}$ in any smallest-to-largest order, and reassign excess tokens from children to their parent.\footnote{We use the usual notions like \emph{children}, \emph{parents}, and \emph{descendants} for the laminar family $\mathcal{L}$.
  More precisely, for $L_1, L_2\in \mathcal{L}$ with $L_2\subsetneq L_1$, we call $L_1$ an \emph{ancestor} of $L_2$, and $L_2$ is a \emph{descendant} of $L_1$.
  If $L_2$ is a descendant of $L_1$ and there is no set $L\in \mathcal{L}$ with $L_2 \subsetneq L \subsetneq L_1$, then $L_2$ is a \emph{child} of $L_1$, and $L_1$ is called the \emph{parent} of $L_2$.%
}
  With a \emph{smallest-to-largest} order, we mean any order such that when considering $L\in \mathcal{L}$, then all descendants of $L$ in $\mathcal{L}$ have already been considered.
  We maintain the following invariant: After considering a set $L\in \mathcal{L}$, we have reassigned the tokens of $L$ and its descendants in a way that $L$ has at least $4$ tokens, and each of its descendants has $2$ tokens.
  By showing that this invariant can be maintained, the results follows because at the end of the procedure, the maximal sets in $\mathcal{L}$ will have obtained at least $4$ tokens, and all other sets in $\mathcal{L}$ obtained $2$ tokens.

  We start by showing that the invariant holds for each minimal set $L$ in $\mathcal{L}$.
  If the edges $\fracpart(y)$ have at least $4$ endpoints in $L$, then the invariant holds.
  Otherwise, we have $|\fracpart(y) \cap E(L,V)|\leq 3$, and because we assumed that we cannot apply a cut relaxation to $L$, there must be a smaller set $S\subseteq L$ corresponding to a $y$-tight $g$-cut constraint.
  We choose $S$ to be a minimal such set.
  However, because $\delta(S) \subseteq E(L,V)$, we have $|\fracpart(y) \cap \delta(S)| \leq  |\fracpart(y)\cap E(L,V)| \leq 3$, i.e., the set $\delta(S)$ contains at most $3$ edges with fractional $y$-values.
  This implies that we could have applied a cut relaxation to $S$, which contradicts our assumption that no cut relaxation was possible.

  Consider now a non-minimal set $L$ in $\mathcal{L}$, and assume that the invariant holds for all of its children.
  If $L$ has at least two children in $\mathcal{L}$, then it can get $2$ tokens from each of them, and the invariant holds for $L$.
  Hence, assume that $L$ has only one child $C\in \mathcal{L}$ in $\mathcal{L}$.
  In this case, $L$ can get two tokens from $C$, which has $4$ tokens.
  We complete the proof by showing that there are at least two edges of $\fracpart(y)$ with one endpoint in $L\setminus C$, which will give an additional $2$ tokens to $L$, to obtain the $4$ tokens required by the invariant.
  First observe that we must have
  \begin{equation*}
    \fracpart(y)\cap \delta(C) \neq \fracpart(y)\cap \delta(L);
  \end{equation*}
  since otherwise the $y$-tight $g$-cut constraint corresponding to $L$ is implied by the $y$-tight $g$-cut constraint that corresponds to $C$ and the $y$-tight equality constraints on single edges.
  This implies
  \begin{equation*}
    \fracpart(y) \cap \delta(L\setminus C) \neq \emptyset,
  \end{equation*}
  and already shows that $L$ obtains at least one more token due to a fractional edge with an endpoint in $L\setminus C$.
  Because both $L$ and $C$ correspond to $y$-tight $g$-cut constraints, we have
  \begin{align}
    y(\delta(L)) &= k - g(\delta(L)), \text{ and}\\
    y(\delta(C)) &= k - g(\delta(C)).
  \end{align}
  Due to integrality of the ghost values $g$, this implies $y(\delta(L))\in \mathbb{Z}$ and $y(\delta(C)) \in \mathbb{Z}$.
  Thus,
  \begin{equation}\label{eq:diffDeltaCDeltaLInt}
    y(\delta(L)) - y(\delta(C)) = y(E(L\setminus C, V\setminus L)) - y(E(L\setminus C, C)) \in \mathbb{Z}.
  \end{equation}
  Note that
  \begin{equation*}
    \delta(L\setminus C) = E(L\setminus C, V\setminus L) \cup E(L\setminus C, C),
  \end{equation*}
  and~\eqref{eq:diffDeltaCDeltaLInt} thus implies that $\delta(L\setminus C)$ cannot have a single edge with fractional $y$-value.
  Because $\fracpart(y)\cap \delta(L\setminus C)\neq \emptyset$, we have $|\fracpart(y)\cap \delta(L\setminus C)| \geq 2$, showing as desired that $L$ gets at least two tokens from edges of $\fracpart(y)$ with one endpoint in $L\setminus C$.
\end{proof}

Note that \cref{lem:progress} readily implies that \cref{alg:main} terminates.
Indeed, the number of ghost value augmentations one can perform on an edge $e$ is no more than $\lceil \sfrac{k}{4} \rceil$, because then just the ghost values alone already provide a load of at least $\sfrac{k}{2}$ on $e$, and $e$ therefore does not qualify anymore for a ghost value augmentation.
Moreover, the sets we relax correspond to the laminar family $\mathcal{R}$, which can have size at most $O(|\overline{V}|)$.

However, bounding the number of ghost value augmentations per edge by $\lceil \sfrac{k}{4} \rceil$ turns out to be very loose.
Actually, as we will show next, we can perform at most one ghost value augmentation per edge, which is a result that is also helpful later on.
This holds because whenever a ghost value augmentation is performed on an edge with endpoints $u$ and $v$, then a large $y$-value on $E(u,v)$ is integral and, because we fix/freeze integral values, we will have a load of at least $\lfloor y \rfloor (E(u,v))$ between the vertices $u$ and $v$ in any future iteration where $u$ and $v$ did not yet get contracted through a cut relaxation.
This load will be too high for another ghost value augmentation to be performed on an edge with endpoints $u$ and $v$.

The following lemma implies in particular that at most one ghost value augmentation can be applied per edge.
\begin{lemma}\label{lem:noPosGhostValInEuv}
  At any iteration of the algorithm when a ghost value augmentation is performed between two vertices $u$ and $v$, then no edge in $E(u,v)$ has strictly positive ghost value, i.e., $g_f=0$ for $f\in E(u,v)$.
\end{lemma}
\begin{proof}
Consider an iteration of the algorithm, with solution $y$ to \LPA and ghost values $g$, and let $e\in E$ be an edge with endpoints $u$ and $v$ and strictly positive ghost value, i.e., $g_e > 0$.
  We show the statement by showing that no ghost value augmentation can be applied in the current iteration to any edge in $E(u,v)$.
Consider a prior iteration when a ghost value augmentation was applied to $e$.
  (Such an iteration must exist because $g_e > 0$.)
  Let $F$ be the set of edges parallel to $e$ at the moment when this ghost value augmentation was applied to $e$, and let $\widetilde{y}$ and $\widetilde{g}$ be the \LPA solution and ghost values, respectively, at that iteration.
  Thus, $\widetilde{g}$ are the ghost values right before the ghost value augmentation to $e$ happened.
Hence, $F \subseteq E(u,v)$, and
  \begin{equation*}
    (\widetilde{y}+\widetilde{g})(F) \geq \frac{k}{2} -2.
  \end{equation*}
  Because ghost values are always integral, $k$ is even by assumption, and, by \cref{lem:atMostOneFracInEuv}, $\widetilde{y}$ has at most one fractional value, we have
  \begin{equation*}
    \lfloor \widetilde{y} + \widetilde{g} \rfloor (F) \geq \left\lfloor \frac{k}{2} -2 \right\rfloor = \frac{k}{2} - 2.
  \end{equation*}
  Moreover, as the ghost value of $e$ gets augmented by two units, and because integral $\widetilde{y}$-values are fixed, we have
  \begin{equation*}
    (y+g)(F) \geq \lfloor \widetilde{y}+\widetilde{g}\rfloor(F) + 2 \geq \frac{k}{2},
  \end{equation*}
  showing as desired that no ghost value augmentation can be applied to any edge in $E(u,v)$ at the current iteration.
\end{proof}

Using the above lemmas, we can now provide a better bound on the number of iterations of \cref{alg:main} than the $\lceil \frac{k}{4} \rceil \cdot |E| + O(|\overline{V}|)$ mentioned above.
This lemma will also be useful again later on.
\begin{lemma}\label{lem:mainAlgIterBound}
  \cref{alg:main} terminates within $O(|\overline{V}|)$ iterations.
\end{lemma}
\begin{proof}
  By \cref{lem:noPosGhostValInEuv}, at most one ghost value augmentation can be performed per edge.
  Hence, to bound the number of iterations, it suffices to bound the number of edges contained in $G=(V,E)$ during the first iteration.
  Note that after~\cref{algline:solveLPInit} of \cref{alg:main}, we have $|\supp(y)| = |E| \leq |\overline{V}|$, because vertex solutions to~\ref{lp:ghost} are sparse, i.e., $|\supp(y)| = O(|\overline{V}|)$.
  This follows from standard combinatorial uncrossing techniques, which imply that the vector $y$ computed during the initialization of the algorithm is the unique solution to a square full rank system corresponding to tight cut constraints for cut sets in a laminar family $\mathcal{L}$.
  Hence, because the system is square, we have $|\supp(y)| \leq |\mathcal{L}|$.
  Finally, a laminar family $\mathcal{L}$ over $\overline{V}$ can have size at most $O(|\overline{V}|)$.
  Thus, the number of ghost value augmentations is bounded by $O(|\overline{V}|)$.

  Moreover, the number of cut relaxations is bounded by $|\mathcal{R}| = O(|\overline{V}|)$, because $\mathcal{R}$ is a laminar family over $\overline{V}$.

  The result now follows because, by \cref{lem:progress}, each iteration of the algorithm either performs a ghost value augmentation, a cut relaxation, or we have that $y$ is integral in which case the algorithm terminates.
\end{proof}

\subsection{Guarantees on Cut Constraints}

Our goal now is to show that the solution $z$ returned by \cref{alg:main} satisfies $z(\delta(S)) \geq k-9$ for all $S\subseteq \overline{V}\setminus \{r\}$ with $S\neq \emptyset$ (again, recall we are assuming that $k$ is even; hence the $9$ instead of $10$).
To achieve guarantees on $z$, we crucially and repeatedly exploit that integral $y$-values get frozen/fixed.
This guarantees that at any iteration of \cref{alg:main}, the solution $y$ to \LPA provides the following lower bound on entries of $z$, i.e., the solution returned by the algorithm.
\begin{equation*}
  \lfloor y \rfloor_{e} \leq z_e \qquad \forall e\in E.
\end{equation*}
We recall that $z$ in the above inequality is the final solution returned by our algorithm (not an intermediate value of $z$).

We therefore start by lower bounding the $\lfloor y \rfloor$-values on different edge sets, starting with edges parallel to an edge to which a ghost value augmentation was applied.
\begin{lemma}\label{lem:largeIntYAtGhostAug}
  At any iteration of \cref{alg:main} where we apply a ghost value augmentation, say on an edge $e$ with endpoints $u$ and $v$, we have $\lfloor y \rfloor (E(u,v)) \geq \frac{k}{2}-2$.
\end{lemma}
\begin{proof}
  By \cref{lem:noPosGhostValInEuv}, right before applying a ghost value augmentation on edge $e$ between vertices $u$ and $v$, we have $g_f=0$ for all $f\in E(u,v)$.
  Because a ghost value augmentation can now be applied to $E(u,v)$, we have $y(E(u,v)) = (y+g)(E(u,v)) \geq \frac{k}{2} -2$.
  Finally, as $y$ has at most one fractional value in $E(u,v)$ due to \cref{lem:atMostOneFracInEuv} and $k$ is even, we obtain as desired
  \begin{equation*}
    \lfloor y \rfloor (E(u,v)) \geq \left\lfloor \frac{k}{2} - 2 \right\rfloor = \frac{k}{2} - 2.\qedhere
  \end{equation*}
\end{proof}

The following lemma shows that $y$-values within a relaxed set are integral.
This shows in particular that our algorithm does return an integral vector $z$.
\begin{lemma}\label{lem:integralInContractedSet}
  Consider any iteration of the algorithm that will relax a cut $S\subseteq V\setminus \{r\}$.
  Then $y_e\in \mathbb{Z}$ for $e\in E[S]$.
\end{lemma}
\begin{proof}
  Let $\mathcal{L} \subseteq 2^{V\setminus \{r\}}$ be a maximal laminar family of $y$-tight $g$-cut constraints that contains the set $S$.
  Consider a reference equation system containing an equation for each $g$-cut constraint for cut sets in $\mathcal{L}$ and all $y$-tight constraints on single edges.
  By \cref{lem:lamTightSys}, this reference system has full column rank and its unique solution is thus $y$.
  Because $S$ will be relaxed, there is no $y$-tight $g$-cut constraint $x(\delta(T))\geq k - g(\delta(S))$ with $T\subsetneq S$ in the reference system.
  Hence, the reference system cannot contain an equation corresponding to any $y$-tight $g$-cut constraint of this type.
  This implies in particular that no edge $e\in E[S]$ is part of an equation in the reference system that corresponds to a $y$-tight $g$-cut constraint.
  However, if there was an edge $e\in E[S]$ with $y_e \not\in \mathbb{Z}$, then such an edge would also not be part of any $y$-tight constraint on single edges, because these are constraints requiring that $y$-values on certain edges are integral, and $y_e$ is not integral.
  This violates that each edge in the support of $y$ must be contained in at least one equation of any full column rank equation system over the edges, thus implying the statement.
\end{proof}

Before making a statement about the $z$-loads on cuts, we observe that we have high $\lfloor y \rfloor$-loads on cuts right before the algorithm either relaxes a set containing them or terminates.
\begin{lemma}\label{lem:cutGoodAfterRelax}
  Consider an iteration of \cref{alg:main}.
  Let $S\subseteq V\setminus \{r\}$ with $S\neq \emptyset$, and assume that at this iteration of the algorithm either a cut $R\subseteq V$ with $S\subseteq R$ gets relaxed, or the algorithm terminates. Then
  \begin{equation*}
    \lfloor y \rfloor (\delta(S)) \geq k - 4.
  \end{equation*}
\end{lemma}
\begin{proof}
  We start by observing that without loss of generality, we can assume that the $g$-cut constraint corresponding to $S$ is still part of \LPA during the considered iteration.
  If this is not the case, then we consider the iteration where $S$ got relaxed, and apply the result to that iteration, during which the above assumption holds.
  If we denote by $\widetilde{y}$ the $y$-values of the iteration when $S$ got relaxed, then we get $\lfloor \widetilde{y} \rfloor (\delta(S)) \geq k-4$, which implies $\lfloor y \rfloor (\delta(S)) \geq k-4$, because integral $y$-values get fixed.

  Hence, we assume from now on that the $g$-cut constraint corresponding to $S$ is still part of \LPA at the considered iteration.
  We make a case distinction based on the number of edges $e\in \delta(S)$ to which a ghost value augmentation was applied, i.e., the number of edges $e\in \delta(S)$ with $g_e = 2$.
  If there are at least two distinct edges $e_1, e_2\in \delta(S)$ with $g_{e_1} = g_{e_2} = 2$, then let $F_j \subseteq \overline{E}$ be the edges that were parallel to $e_j$, for $j\in \{1,2\}$, when the ghost value augmentation was applied to $e_j$.
  \cref{lem:noPosGhostValInEuv} implies $F_1\cap F_2 = \emptyset$.
  Moreover, by \cref{lem:largeIntYAtGhostAug} and the fact that integral $y$-values get fixed, we have $\lfloor y \rfloor (F_j) \geq \frac{k}{2}-2$ for $j\in \{1,2\}$.
  Hence, $\lfloor y \rfloor (\delta(S)) \geq \lfloor y \rfloor (F_1) + \lfloor y \rfloor (F_2) \geq k - 4$, as desired.

  Assume now that a ghost value augmentation was applied to at most one edge in $\delta(S)$, i.e., $g(\delta(S)) \leq 2$.
  Because the $g$-cut constraint corresponding to $S$ is still part of \LPA, we get
  \begin{equation}\label{eq:gCutStillValidInCutGoodAfterRelax}
    (y+g)(\delta(S)) \geq k.
  \end{equation}
  If the algorithm terminates at this moment, we have that $y$ is integral, and hence
  \begin{equation*}
    \lfloor y \rfloor (\delta(S)) = y(\delta(S)) = k - g(\delta(S)) \geq k - 2,
  \end{equation*}
  because $g(\delta(S)) \leq 2$.

  Otherwise, we are in the case where a set $R\supseteq S$ gets relaxed in the current iteration.
  Note that $\delta(S) \subseteq \delta(R) \cup E[R]$.
  Moreover, because $|\fracpart(y) \cap \delta(R)|\leq 3$, and $y_e\in \mathbb{Z}$ for all $e\in E[R]$ due to \cref{lem:integralInContractedSet}, we have $|\fracpart(y)\cap \delta(S)|\leq 3$.
  Together with~\eqref{eq:gCutStillValidInCutGoodAfterRelax}, this implies $(\lfloor y \rfloor + g)(\delta(S)) > k - 3$.
  Because the left-hand side is integral, we get
  \begin{equation*}
  (\lfloor y \rfloor + g)(\delta(S)) \geq k-2,
  \end{equation*}
  which implies
  \begin{equation*}
    \lfloor y \rfloor \geq k - 2 - g(\delta(S)) \geq k - 4,
  \end{equation*}
  because $g(\delta(S)) \leq 2$.
\end{proof}

Whereas the previous statement provided guarantees for the current solution $y$ at some iteration of the algorithm, we now derive from this guarantees for the final solution $z$ returned by \cref{alg:main}.
Recall that $\mathcal{R}$ is the set family whose sets correspond to non-singleton sets we relaxed (by contraction and dropping the corresponding cut from \LPA); see the beginning of \Cref{sec:proof} for a definition.
We start by providing a guarantee for cut sets $S\in \overline{V}\setminus \{r\}$ that are ``compatible'' with the laminar family $\mathcal{R}$, i.e., adding $S$ to $\mathcal{R}$ preserves laminarity. 
\begin{lemma}\label{lem:lamCompCutsGood}
  For any $S\subseteq \overline{V}\setminus \{r\}$ such that $\mathcal{R}\cup \{S\}$ is laminar, we have
  \begin{equation*}
    z(\delta(S)) \geq k - 4.
  \end{equation*}
\end{lemma}
\begin{proof}
  If there is no set $R\in \mathcal{R}$ with $S\subseteq R$, then the $g$-cut constraint corresponding to $S$ is in \LPA until the last iteration, and we can apply \cref{lem:cutGoodAfterRelax} to the last iteration of the algorithm to obtain $\lfloor y \rfloor (\delta(S)) \geq k-4$.
  The result follows because the $z$-values on the edges $\delta(S)$ are the same as the $y$-values during the last iteration.

  Otherwise, let $R\in \mathcal{R}$ with $S\subseteq R$ be the smallest set in $\mathcal{R}$ containing $S$, and consider the iteration where $R$ gets relaxed.
  By \cref{lem:cutGoodAfterRelax}, we obtain $\lfloor y \rfloor (\delta(S)) \geq k-4$, and the result follows because integer values get fixed, and hence $z \geq \lfloor y \rfloor$.
  Indeed, this implies $z(\delta(S)) \geq \lfloor y \rfloor (\delta(S)) \geq k-4$.
\end{proof}

We now turn to showing that $z$ has a large load on any cut set $S\subseteq \overline{V}\setminus \{r\}$, even if $S$ is not compatible with $\mathcal{R}$.
Such a non-compatible set $S$ will cross at least one set $R\in \mathcal{R}$.
We first show that there is a large $z$-value on edges between $R\setminus S$ and $R\cap S$.
(The set $S$ in the statement below corresponds to the set $S\cap R$ when $S$ is a cut set crossing $R$.)

\begin{lemma}\label{lem:kOver2InsideRelaxed}
  Let $R\in \mathcal{R}$ and $S\subsetneq R$ with $S\neq \emptyset$.
  Then $z(\overline{E}(S \cap R ,R\setminus S)) \geq \frac{k}{2} - 4$.
\end{lemma}
\begin{proof}
  We first observe that we can assume the following minimality condition for the pair $R$ and $S$: there is no set $\widetilde{R}\in \mathcal{R}$ with $\widetilde{R} \subsetneq R$ so that $S$ and $\widetilde{R}$ are crossing.
  Assume that we proved the statement only for such sets $R$ and $S$ fulfilling this minimality condition.
  Given an arbitrary $R\in \mathcal{R}$ and $S\subsetneq R$ with $S\neq \emptyset$, we define $\widetilde{R}$ to be the smallest set $\widetilde{R} \in \mathcal{R}$ contained within $R$ that crosses $S$.
  Moreover, we set $\widetilde{S} \coloneqq S\cap \widetilde{R}$.
  This pair $\widetilde{R}$ and $\widetilde{S}$ fulfills the above-mentioned minimality condition.
  Hence,
  \begin{equation*}
    z(\overline{E}(\widetilde{S}, \widetilde{R}\setminus \widetilde{S})) \geq \frac{k}{2}-4.
  \end{equation*}
  Moreover, because $\overline{E}(\widetilde{S}, \widetilde{R} \setminus \widetilde{S}) \subseteq \overline{E}(S, R\setminus S)$, we also get $z(\overline{E}(S, R\setminus S)) \geq \frac{k}{2} - 4$, as desired.

  Hence, assume from now on that $R$ and $S$ fulfill the above-mentioned minimality condition.
  This minimality condition implies
  \begin{enumerate}[nosep]
    \item\label{item:RSLaminar} $\mathcal{R} \cup \{S\}$ is laminar, and
    \item\label{item:RNoSLaminar} $\mathcal{R} \cup \{R\setminus S\}$ is laminar.
  \end{enumerate}
  Consider now the iteration of the algorithm when the set $R$ gets relaxed.
  Due to~\eqref{item:RSLaminar} and~\eqref{item:RNoSLaminar}, there are vertex sets $A,B \subseteq V\setminus \{r\}$ of the current graph $G=(V,E)$ that correspond to $S$ and $R\setminus S$, respectively, i.e., $\delta(A) = \delta(S)$ and $\delta(B) = \delta(R\setminus S)$.
  In other words, when undoing all contractions within $A$ we get $S$, and when undoing all contractions within $B$ we get $R\setminus S$.
  Hence, by \cref{lem:cutGoodAfterRelax}, we have
  \begin{equation}\label{lem:lowerBoundSAndRNoS}
  \begin{aligned}
    y(\delta(S)) = y(\delta(A)) &\geq k-4, \text{ and} \\
    y(\delta(R\setminus S)) = y(\delta(B)) &\geq k-4.
  \end{aligned}
  \end{equation}
  The desired result now follows from:
  \begin{align*}
    z(\overline{E}(S, R\setminus S)) &= z(E(S, R\setminus S)) \\
    &= y(E(S, R\setminus S)) \\
    &= \frac{1}{2}\left( y(\delta(S)) + y(\delta(R\setminus S)) - y(\delta(R)) \right) \\
    &\geq \frac{k}{2} - 4.
  \end{align*}
  The first equation holds because at the iteration when $R$ gets contracted, no edge with one endpoint in $S$ and one in $R\setminus S$ has been contracted yet.
  The second equation holds because when $R$ gets relaxed, \cref{alg:main} sets the $z$-values to the $y$-values for edges in $E[R]$.
  Finally, the inequality follows from~\eqref{lem:lowerBoundSAndRNoS} and $y(\delta(R)) \leq (y+g)(\delta(R)) = k$, which holds because $R$ corresponds to a $y$-tight $g$-cut constraint as it gets relaxed at the considered iteration.
\end{proof}

Before proving our final guarantee for $z$-values of any cut set $S\subseteq \overline{V}\setminus \{r\}$ with $S\neq \emptyset$, we show that the edges crossing a set $R \in \mathcal{R}$ will have a $z$-load of at most $k + 2$.
This is the last ingredient we need to provide a lower bound on the $z$-value of the edges crossing any cut set.
\begin{lemma}\label{lem:zUpperBound}
  For any $R \in \mathcal{R}$, we have
  \begin{equation*}
    z(\delta(R)) \leq k + 2.
  \end{equation*}
\end{lemma}
\begin{proof}
  Consider the iteration when $R$ gets relaxed.
  At this iteration we have $(y+g)(\delta(R)) = k$, because the $g$-cut constraint corresponding to $R$ is $y$-tight. 
  
  If $|\fracpart(y) \cap \delta(R)| = 0$ then in any future iteration with \LPA solution $\widetilde{y}$ we have $\widetilde{y}(\delta(R)) \leq k$.
  
  On the other hand, suppose $|\fracpart(y) \cap \delta(R)| > 0$ in the iteration in which we relax $R$. It follows that $\lfloor y \rfloor(\delta(R)) \leq k - 1$. However, since $R$ is about to be relaxed we know $|\fracpart(y) \cap \delta(R)| \leq 3$ and so it follows that in any future iteration with \LPA solution $\widetilde{y}$, we have 
  \begin{align*}
      (\widetilde{y})(\delta(R)) \leq \lfloor y \rfloor(\delta(R)) + |\fracpart(y) \cap \delta(R)| \leq k + 2.
  \end{align*}
  Thus, $z(\delta(R)) \leq k+2$.
\end{proof}

Finally, \cref{lem:cutLB} shows that $z$ has a large load on the edges crossing any cut set $S\subseteq \overline{V}\setminus \{r\}$.
More precisely, the load will be at least $k-9$ (for $k$ even), as claimed by our main rounding theorem, \cref{thm:mainaux}.
\begin{lemma}\label{lem:cutLB}
For all $S\subseteq \overline{V}\setminus \{r\}$ where $S \neq \emptyset$ we have
  \begin{equation*}
    z(\delta(S)) \geq k - 9.
  \end{equation*}
\end{lemma}
\begin{proof}
  The result follows from \cref{lem:lamCompCutsGood} if $\mathcal{R}\cup \{S\}$ is laminar.
  Hence, assume from now on that at least one set $R\in \mathcal{R}$ crosses $S$.

  If there are at least two disjoint sets $R_1,R_2\in \mathcal{R}$ that cross $S$, then we have
  \begin{equation*}
    z(\delta(S)) \geq z(\overline{E}(R_1\cap S, R_1\setminus S)) + z(\overline{E}(R_2\cap S, R_2\setminus S)) \geq 2 \cdot \left(\frac{k}{2} - 2\right) = k-4,
  \end{equation*}
  by \cref{lem:kOver2InsideRelaxed}.

  Hence, assume in what follows that there are no two disjoint sets in $\mathcal{R}$ that cross $S$.
  Let $R\in \mathcal{R}$ be the maximal set in $\mathcal{R}$ that crosses $S$.
  Maximality of $R$ and the fact that there is no set in $\mathcal{R}$ disjoint from $R$ that crosses $S$ implies:
  \begin{itemize}[nosep]
    \item $\mathcal{R} \cup \{R \cup S\}$ is laminar, and
    \item $\mathcal{R} \cup \{S\setminus R\}$ is laminar.
  \end{itemize}
  By a simple counting of edges leaving the relevant sets, we have
  \begin{align}
    z(\delta(S)) &\geq z(\overline{E}(S\cap R, R\setminus S)) + \frac{1}{2}\Big( z(\delta(S\setminus R)) + z(\delta(S\cup R)) - z(\delta(R))\Big).
  \end{align}
  The result now follows by using
  \begin{itemize}[nosep]
    \item $z(\overline{E}(S\cap R, R\setminus S)) \geq \frac{k}{2} - 4$ by \cref{lem:kOver2InsideRelaxed},
    \item $z(\delta(S\setminus R)) \geq k - 4$ by \cref{lem:lamCompCutsGood}, which applies due to laminarity of $\mathcal{R} \cup \{R\cup S\}$,
    \item $z(\delta(S\cup R)) \geq k -4$ by \cref{lem:lamCompCutsGood}, which applies due to laminarity of $\mathcal{R}\cup \{S\setminus R\}$, and
    \item $z(\delta(R)) \leq k+2$, which holds because of \cref{lem:zUpperBound}. \qedhere
  \end{itemize}
\end{proof}

\subsection{Cost of Returned Solution}

The required cost bound is easily proven.
\begin{observation}\label{obs:costBound}
  The solution $z$ returned by \cref{alg:main} has cost bounded by $c^{\top} y$ where $y$ is the input vector of \Cref{thm:mainaux}.
\end{observation}
\begin{proof}
  We denote by $\LPOPTkGhost[k]$ the optimal value of the \LPA solution computed in \cref{algline:solveLPInit} of \cref{alg:main}. Observe that
  \begin{align*}
    \LPOPTkGhost[k] \leq c^{\top} y
  \end{align*}
  since $y$ is feasible for this LP (where again, $y$ is the input to \Cref{thm:mainaux}).
  Whenever we change \LPA in \cref{alg:main}, we do so by either increasing ghost values in \cref{algline:ghostValueAugm}, performing a cut relaxation in \cref{algline:relaxation}, or by adding equality constraints in \cref{algline:addEqualityConstr}.
  All of these operations have the property that the previous LP solution is still valid for the new LP.
  (In case of a cut relaxation, we can use the values of the previous LP solution on the non-contracted edges.)
  This implies that the returned solution $z$ fulfills $c^{\top} z \leq \LPOPTkGhost[k] \leq c^{\top} y$, as desired.
\end{proof}

\subsection{Polynomial Running Time and Putting Things Together}
Lastly, we prove the polynomial runtime of our algorithm and then observe that, by the above analysis, we can conclude \Cref{thm:mainaux}.

\begin{lemma}\label{lem:solveLPEfficient}
  Determining a vertex solution to any of the linear programs solved in \cref{alg:main}, i.e., during initialization in \cref{algline:solveLPInit} and later in the iterations in \cref{algline:resolveLP}, can be done in poly-time.
\end{lemma}
\begin{proof}
  Note that each linear program \LPA we have to solve in \cref{alg:main} are of the following form.
  We are given a graph $G=(V,E)$ with ghost values $g\in \{0,2\}^E$.
  The graph $G$ was obtained from the original graph $\overline{G}=(\overline{V},\overline{E})$ by deleting edges and contracting some vertex sets.
  Hence, a vertex in $G$ may correspond to a set of vertices of $\overline{G}$.
  Recall that, for the vertices $W\subseteq V$ of $G$ that were obtained through a cut relaxation, the LP does not have a lower bound of $k$ on their degree anymore.
  Moreover, we have lower bounds ($\ell \leq x$) and upper bounds ($x\leq u$) on edge variables because the LP we consider in \cref{alg:main} stems from \eqref{lp:ghost}, and also because \cref{algline:addEqualityConstr} of \cref{alg:main} adds equality constraints, each of which corresponds to a matching lower and upper bound for an edge variable $x_e$.
  Thus, the LPs we have to solve are of the following form:
\begin{equation}\label{eq:lpToSolve}
\begin{array}{rrcll}
  \min &c^\top x     &     &                       &                                                     \\
       &x(\delta(S)) &\geq &k - g(\delta(S))       &\forall S\subseteq V\setminus \{r\}, |S|\geq 2 \\
       &x(\delta(v)) &\geq &k - g(\delta(v))       &\forall v\subseteq V\setminus (W\cup \{r\}) \\
       &x            &\in  &[\ell,u]^E.            &
\end{array}
\end{equation}
\cref{alg:main} needs to find an optimal vertex solution to this LP in polynomial time.

One conceptually easy way to achieve this is through the Ellipsoid Method.
  Because the encoding size of each of our constraints is small, an optimal vertex solution to~\eqref{eq:lpToSolve} can be obtained in polynomial time if we can separate in polynomial time.
  (We refer the interested reader to the excellent book of \textcite{groetschel_1981_ellipsoid} for details on the Ellipsoid Method and separation.)
  To solve the separation problem, we are given a point $y\in \mathbb{R}^E_{\geq 0}$, and we show that we can determine in poly-time whether $y$ is feasible for~\eqref{eq:lpToSolve} and, if this is not the case, return a violated constraint of~\eqref{eq:lpToSolve}.
  
  This easily reduces to a polynomial number of minimum $s$-$t$ cut computations as follows.
  First, we verify that $y$ fulfills the lower and upper bounds, i.e., $y\geq \ell$ and $y\leq u$.
  We then determine for each $v\in V\setminus (W\setminus \{r\})$ a minimum $v$-$r$ cut, when using $y+g$ as edge capacities.
  If any of these cuts has a value below $k$, then it corresponds to a violated constraint of~\eqref{eq:lpToSolve}.
  More precisely, these cuts allow for verifying all constraints of~\eqref{eq:lpToSolve} of the form $x(\delta(v)) \geq k - g(\delta(v))$ for $v\in V\setminus (W\setminus \{r\})$, and all constraints of the first constraint family of~\eqref{eq:lpToSolve} for any $S\subseteq V\setminus \{r\}$ containing at least one vertex of $V\setminus W$.
  For the remaining constraints, which correspond to constraints of the first constraint family of~\eqref{eq:lpToSolve} for sets $S\subseteq W$ with $|S|\geq 2$, we compute, for each subset $Q\subseteq W$ with $|Q|=2$, a minimum $Q$-$r$ cut with edge capacities given by $x+g$.
  This allows for checking all remaining constraints and for identifying a violated constraint if there is one, which completes the separation procedure.%
\footnote{Alternatively, to obtain in polynomial time an optimal vertex solution to~\eqref{eq:lpToSolve}, one can also use a compact extended formulation.
This can, for example, be achieved by having a set of variables $x\in \mathbb{R}^E_{\geq 0}$ that correspond to installing capacities on the edges (for each $e\in E$, we get a capacity of $g(e)$ for free), and then capturing the $s$-$t$ cut conditions we used in the separation problem through flow problems with respect to the capacities $x+g$.
  For each of the (quadratically many) $s$-$t$ cut computations, we use an independent set of flow variables and require that the corresponding $s$-$t$ flow has a value of at least $k$.
  }
\end{proof}

\begin{lemma}\label{lem:findCutToRelaxEfficient}
  At any iteration of \cref{alg:main}, one can in poly-time determine a set $S \subseteq V\setminus \{r\}$ corresponding to a $y$-tight $g$-cut constraint to which the cut relaxation operation (\Cref{algline:relaxation}) can be applied, if there is such a set.
\end{lemma}
\begin{proof}
  Consider an iteration of the algorithm and let $G=(V,E)$ be the current graph and $y$ the current solution to \LPA, as usual.
  Because a cut relaxation always contracts the corresponding vertex set, the $g$-cut constraints in the LP contain a constraint for each set $S$ of size at least $2$, and also for a subset of the singleton sets---namely those that do not correspond to an original vertex in $\overline{V}$ and are therefore not a set obtained through contractions.
  Let $W\subseteq V$ be all singletons that correspond to contracted vertices.
  Hence, the family of all $g$-cut constraints of the LP is the following:
\begin{equation}\label{eq:lpForCutRelax}
\begin{array}{rrcll}
       &x(\delta(S)) &\geq &k - g(\delta(S))       &\forall S\subseteq V\setminus \{r\}, |S|\geq 2 \\
       &x(\delta(v)) &\geq &k - g(\delta(v))       &\forall v\subseteq V\setminus (W\cup \{r\}). \\
\end{array}
\end{equation}

  We now first describe the procedure we use to identify a set $S \subseteq V\setminus \{r\}$ to which the cut relaxation operation can be performed (if there is one) and then discuss why it is correct.

  We start by checking whether there is any singleton set $S=\{v\}$ for $v\in V\setminus (W\cup \{r\})$ that can be relaxed.
  This is the case if $x(\delta(v)) = k - g(\delta(v))$ and $|\fracpart(y)\cap \delta(v)| \leq 3$, which can be checked in poly-time.
  If no such singleton exists, we consider all sets $U\subseteq \fracpart(y)$ with $|U| \leq 3$.
  For each such $U$, consider all sets $W\subseteq V(U)$, where $V(U)\subseteq V$ are all endpoints of edges in $U$, such that $W$ contains exactly one endpoint of each edge of $U$.
  If $|W| < 2$, consider all sets $Z\subseteq V$ with $W\subseteq Z$ and $|Z|=2$; otherwise, we only consider $Z=W$ for the set $W$.
  Let $\mathcal{Z}$ be all considered sets $Z$.
  Note that $|\mathcal{Z}| = O(|V|^3)$ because $|Z|\leq 3$.

  For each $Z \in \mathcal{Z}$ we do the following.
  We compute a minimal minimum $Z$-$(\{r\} \cup (V(\fracpart(y))\setminus Z))$ cut in $G$ with edge capacities $y+g$.\footnote{Even though we do not need this fact, it is well known that minimal minimum $A$-$B$ cuts for any disjoint non-empty vertex sets $A$ and $B$ are unique.
  Moreover, the minimal minimum $A$-$B$ cut can for example be obtained by first computing a maximum $A$-$B$ flow $f$, and then choosing the $A$-$B$ cut that consists of all vertices reachable from $A$ in the $f$-residual graph.
  This step of determining vertices reachable from $A$ can be performed in linear time, for example through breadth-first search, once a maximum $A$-$B$ flow $f$ has been computed.
  This is one way to see that this step of the procedure can be performed in poly-time.
  }
  Let $S_Z \subseteq V\setminus \{r\}$ be the computed cut for $Z\in \mathcal{Z}$.
  We now focus only on the minimal cuts within the family ${\{S_Z\}}_{Z\in \mathcal{Z}}$.
  If there is a minimal cut $S_Z$ in the family with $(y+g)(\delta(S_Z))=k$, then we claim that a cut relaxation can be applied to $S_Z$.
  If there is no such cut, we claim that no cut relaxation can be applied at the current iteration.
  Note that this procedure clearly runs in polynomial time because $|\mathcal{Z}| = O(|V|^3)$.
  It remains to show its correctness.

  To this end, we first observe that any set identified for relaxation by our procedure is indeed one corresponding to a $g$-cut constraint that can be relaxed.
  This is clear for singleton cuts.
  Hence, assume that there is no singleton cut that can be relaxed, and let $S_Z$ for $Z\in \mathcal{Z}$ be a set identified by our procedure.
  For the sake of deriving a contradiction, assume that there is a set $S\subsetneq S_Z$ that corresponds to a $y$-tight $g$-cut constraint of the LP.
  Let $U\coloneqq \fracpart(y)\cap \delta(S)$ and $W \coloneqq V(U)\cap S$.
  Because $S$ is not a singleton, there is a set $Z\subseteq S$ with $W\subseteq Z$ and $|Z|\geq 2$.
  However, then $S$ is a candidate for the cut problem that determined $S_{Z}$.
  This contradicts the minimality of $S_{Z}$.

  It remains to show that if there is a set $S$ corresponding to a $g$-cut constraint that can be relaxed, then our procedure identifies such a set.
  Note that $U\coloneqq \fracpart(y) \cap \delta(S)$ satisfies $|U|\leq 3$, because the $g$-cut constraint corresponding to $S$ can be relaxed.
  Let $Z\subseteq S$ be any set with
  \begin{enumerate*}[label=(\roman*)]
    \item $V(U)\cap S \subseteq Z$, and
    \item $2 \leq |Z|\leq 3$.
  \end{enumerate*}
  Because $\fracpart(y)\cap E[S]=\emptyset$ by \cref{lem:integralInContractedSet}, the set $S$ is a $Z$-$(\{r\}\cup (V(\fracpart(y))\setminus Z))$ cut.
  Moreover, it corresponds to a $y$-tight $g$-cut constraint that is part of the LP, because one can relax the $g$-cut constraint corresponding to $S$ by assumption.
  Hence, $S$ is a candidate cut for the minimal minimum $Z$-$(\{r\}\cup (V(\fracpart(y))\setminus Z))$ cut $S_{Z}$.
  Thus, the $g$-cut constraint corresponding to $S_Z$ is $y$-tight, because $y(\delta(S_Z))$ cannot be any larger than $y(\delta(S))=k$ as $S_Z$ is a minimum $Z$-$(\{r\}\cup (V(\fracpart(y))\setminus Z))$ cut.
  By the minimality of $S_Z$, we must have $S_Z\subseteq S$.
  Moreover, we cannot have $S_Z \subsetneq S$, because this would imply that there is a $y$-tight $g$-cut constraint corresponding to a strict subset of $S$, which is impossible because the $g$-cut constraint corresponding to $S$ can be relaxed.
  Thus, $S_Z = S$, and our procedure therefore identifies a $g$-cut constraint for the cut relaxation operation.
\end{proof}

\begin{theorem}\label{thm:runtime}
  \cref{alg:main} runs in polynomial time.
\end{theorem}
\begin{proof}
  By \cref{lem:mainAlgIterBound}, the number of iterations of \cref{alg:main} is bounded by $O(|\overline{V}|)$.
  The statement then follows by observing that the initialization and also each iteration of the algorithm can be performed in polynomial time, which holds because of \cref{lem:solveLPEfficient,lem:findCutToRelaxEfficient}.
\end{proof}

Finally, \cref{thm:mainaux} follows from the fact that \cref{alg:main} terminates in polynomial time by \cref{thm:runtime} and the fact that the vector $z$ satisfies $c^{\top} z \leq c^{\top} y$ by \Cref{obs:costBound}, $z_e \in \{\lfloor y_e \rfloor, \lceil y_e\rceil\}$ for every $e$ by construction and $z(\delta(S)) \geq k-9$ for all non-empty $S \subsetneq V$ by \Cref{lem:cutLB}.

\section{Hardness of Approximation of \texorpdfstring{$k$}{k}-ECSM}\label{sec:hardness}
\textcite{Pri11} showed \APX-hardness for $2$-ECSM by showing it is essentially the same problem as the metric version of $2$-ECSS but left  hardness of $k$-ECSM for $k>2$ and as a function of $k$ open, stating
\begin{quote}\centering
    \textit{What [\APX-hardness of $2$-ECSM] leaves to be desired is hardness for $k$-ECSM, $k>2$ and asymptotic dependence on $k$. Why is it hard to show these problems are hard?\ldots A new trick seems to be needed to get a good hardness result for k-ECSM.}
\end{quote}
In this section, we provide such a trick to show $1+ \Omega(\sfrac{1}{k})$ hardness for $k$-ECSM. Specifically, we reduce from the unweighted tree augmentation problem (TAP) problem, which is known to be $(1+ \epsTAP)$-hard-to-approximate for some constant $\epsTAP > 0$. See \Cref{sec:TAPBackground} for details on unweighted TAP. We recall below the hardness statement we show in this section.
\hardness*
The above hardness shows the tightness of \Cref{thm:mainECSS} up to constants.
We prove \cref{thm:hardness} by showing hardness for all odd $k$.
The basic idea is to reduce from unweighted TAP by using the following trick: if a $k$-ECSM solution on graph $G'$ has value exactly $k$ on a node $w$ with two incident edges in $G'$ and $k$ is odd, then one can always make one edge incident to $w$ have value $\lceil k/2 \rceil$ and one edge have value $\lfloor k/2 \rfloor$ while preserving the feasibility of the $k$-ECSM solution.

\subsection{Preliminary: the Unweighted Tree Augmentation Problem (TAP)}\label{sec:TAPBackground}
We give relevant background on the unweighted tree augmentation problem (TAP). We first define unweighted TAP (see \Cref{sfig:tap1}).

\begin{definition}[Unweighted tree augmentation problem (TAP)]\label{dfn:tap}
An instance of the unweighted tree augmentation problem (TAP) consists of a spanning tree $G=(V,E)$ and a set of ``links'' $L \subseteq {V \choose 2}$. A solution is a set of links $F \subseteq L$ such that $(V, E \cup F)$ is $2$-edge-connected. The goal is to minimize $|F|$.
\end{definition}
Feasible solutions to TAP are often equivalently defined in at least two other ways. Say that link $\ell \in L$ \emph{covers} $e \in E$ if $\ell$ is in the fundamental cycle of $e$ in $G$.
We denote by $\cov(e) \subseteq L$ the set of links covering $e$.
Then $F\subseteq L$ is a feasible TAP solution if and only if every edge of $G$ is covered by some link in $F$, i.e., $\cov(e)\cap F\neq\emptyset$ for all $e\in E$.
Similarly, $F$ is a feasible TAP solution if and only if every ``tree cut'' of $G$ (i.e., a cut $S$ satisfying $|\delta_G(S)| = 1$) has at least one edge of $F$ incident to it. We emphasize that a \emph{link} will always refer to an edge of $L$ (and never to e.g.~an edge of $G$). See  \Cref{sfig:tap2} for a feasible TAP solution and \Cref{sfig:tap3} for a tree cut. We let $\OPTTAP$ stand for the cardinality of the minimum cardinality feasible $F$.

The following summarizes known hardness of approximation for unweighted TAP. Notably, we will use the fact that this hardness holds even when the optimal TAP solution has cost at least linear in the number of edges of the input spanning tree.
\begin{theorem}[\APX-Hardness of unweighted TAP,~\cite{kortsarz2004hardness}]\label{thm:tapAPXHard}
There is a constant $\epsTAP \in (0,1)$ such that it is \NP-hard to $(1+\epsTAP)$-approximate unweighted TAP on $G = (V,E)$ even when $\OPTTAP \geq \frac{1}{4} \cdot |E|$.
\end{theorem}
We note two things regarding the above result. First, the authors prove hardness for a slightly different problem than unweighted TAP (namely, \emph{VCAP}, where the goal is to augment to a \emph{$2$-vertex-connected} graph, with all edges having cost $1$ or $2$). It is easy to see that the same reduction works for unweighted TAP. Second, the assumption that $\OPTTAP \geq \frac{1}{4} \cdot |E|$ is not explicitly stated in this work but follows from the fact that $G$ in the unweighted TAP instances have diameter $5$ and one can reduce their diameter to $4$ by contracting a single edge (the edge $\{\bar{r}, r \}$ in the construction presented in~\cite{kortsarz2004hardness} on page~715, which is needed for the vertex-connectivity version but not the edge-connectivity version we have in TAP) without affecting the validity of the reduction. A diameter bound of $4$ on an unweighted TAP instance on tree $G = (V,E)$ immediately implies $\OPTTAP \geq \frac{1}{4} \cdot |E|$ on this instance.

\subsection{Hardness of \texorpdfstring{$k$}{k}-ECSM via TAP Reduction Using Odd \texorpdfstring{$k$}{k}}

    We begin by describing our reduction from unweighted TAP to $k$-ECSM that we will use. Consider an instance of unweighted TAP given by a spanning tree $G = (V,E)$ and links $L \subseteq {V \choose 2}$. Informally, to construct our instance of $k$-ECSM, we begin with our TAP instance and then convert each edge of our TAP instance into two parallel and bisected edges. See \Cref{sfig:tap4}. More formally, we construct the following unit-cost instance of $k$-ECSM on an auxiliary graph $H=(W,B)$.
\begin{enumerate}
  \item \textbf{Vertices:} Let $V_E \coloneqq {\{w_e, w_e'\}}_{e \in E}$ be a set of vertices, two for each edge of $E$. The vertex set of our $k$-ECSM instance is $W \coloneqq V \cup V_E$.
	\item \textbf{Edges:} For each edge $e = \{u, v\} \in E$, we have $4$ edges in our $k$-ECSM instance, namely $\{u, w_e\}$, $\{w_e, v\}$, $\{u, w_e'\}$, and $\{w_e', v\}$. Let $\gadgetEdges$ be all such edges.
    The edge set of our $k$-ECSM instance is $B\coloneqq \gadgetEdges \cup L$.
  \item \textbf{Costs:} The cost of each edge $b\in B$ in our $k$-ECSM instance is $1$, i.e., $c_b =1$.
\end{enumerate}
If $S$ is a tree cut of $G$, then we say that the following cut in $G'$
\emph{corresponds} to $S$:
\begin{equation*}
S' \coloneqq S \cup \{w_e \colon e \in E[S]\} \cup \{w'_e \colon e \in E[S]\},
\end{equation*}
where $E[S]\subseteq E$ denotes all edges of $E$ with both endpoints in $S$.
In words, $S'$ consists of $S$ and all gadget vertices created by edges internal to $S$. See \Cref{sfig:tap6}.

\begin{figure}[h]
	\centering
	 \begin{subfigure}[b]{.32\textwidth}
		\centering
		\includegraphics[width=\textwidth, trim=0mm 0mm 230mm 0mm, clip ]{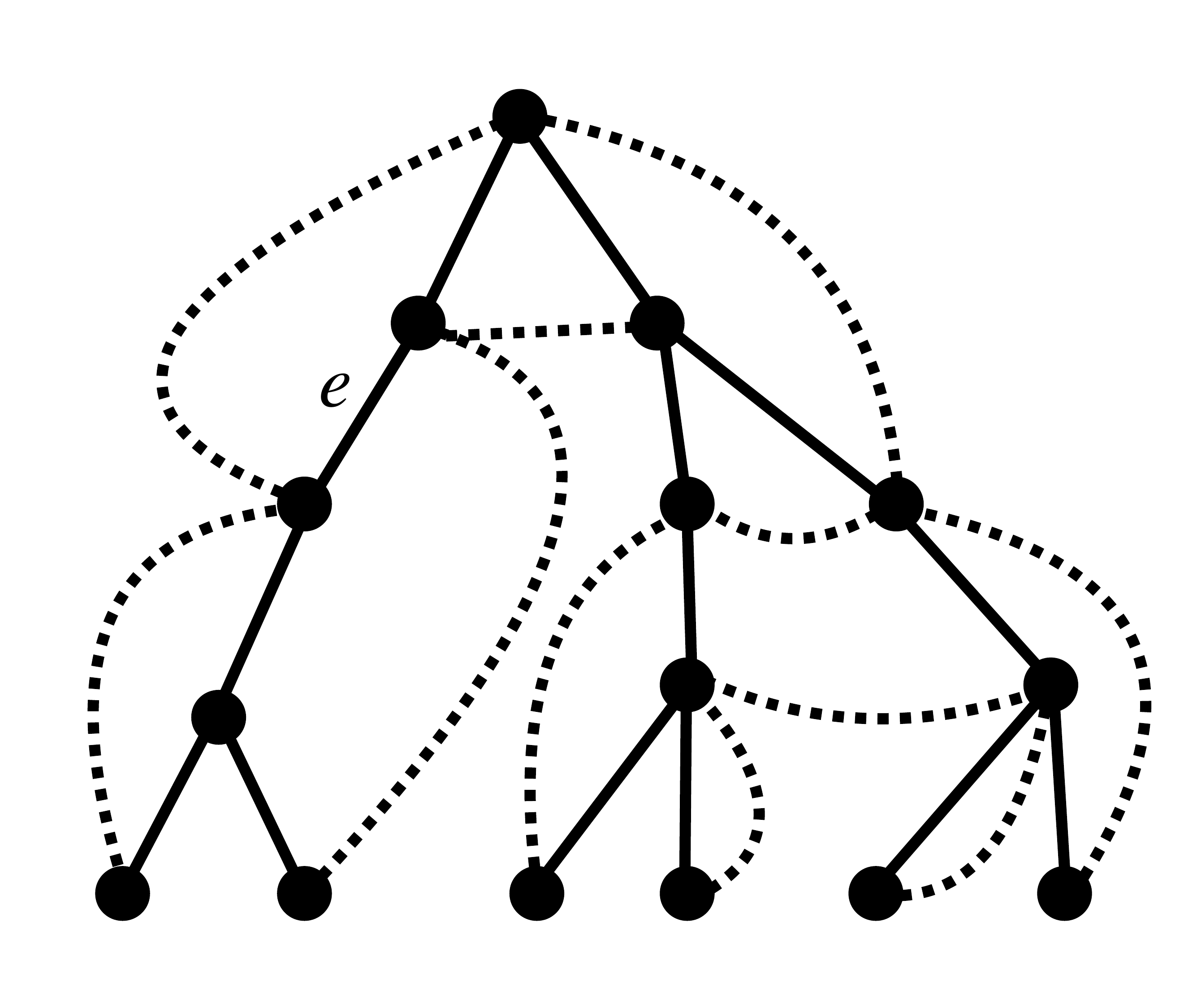}
		\caption{TAP Instance on $G$.}\label{sfig:tap1}
	\end{subfigure}
	\hfill
		 \begin{subfigure}[b]{.32\textwidth}
		\centering
		\includegraphics[width=\textwidth, trim=0mm 0mm 230mm 0mm, clip ]{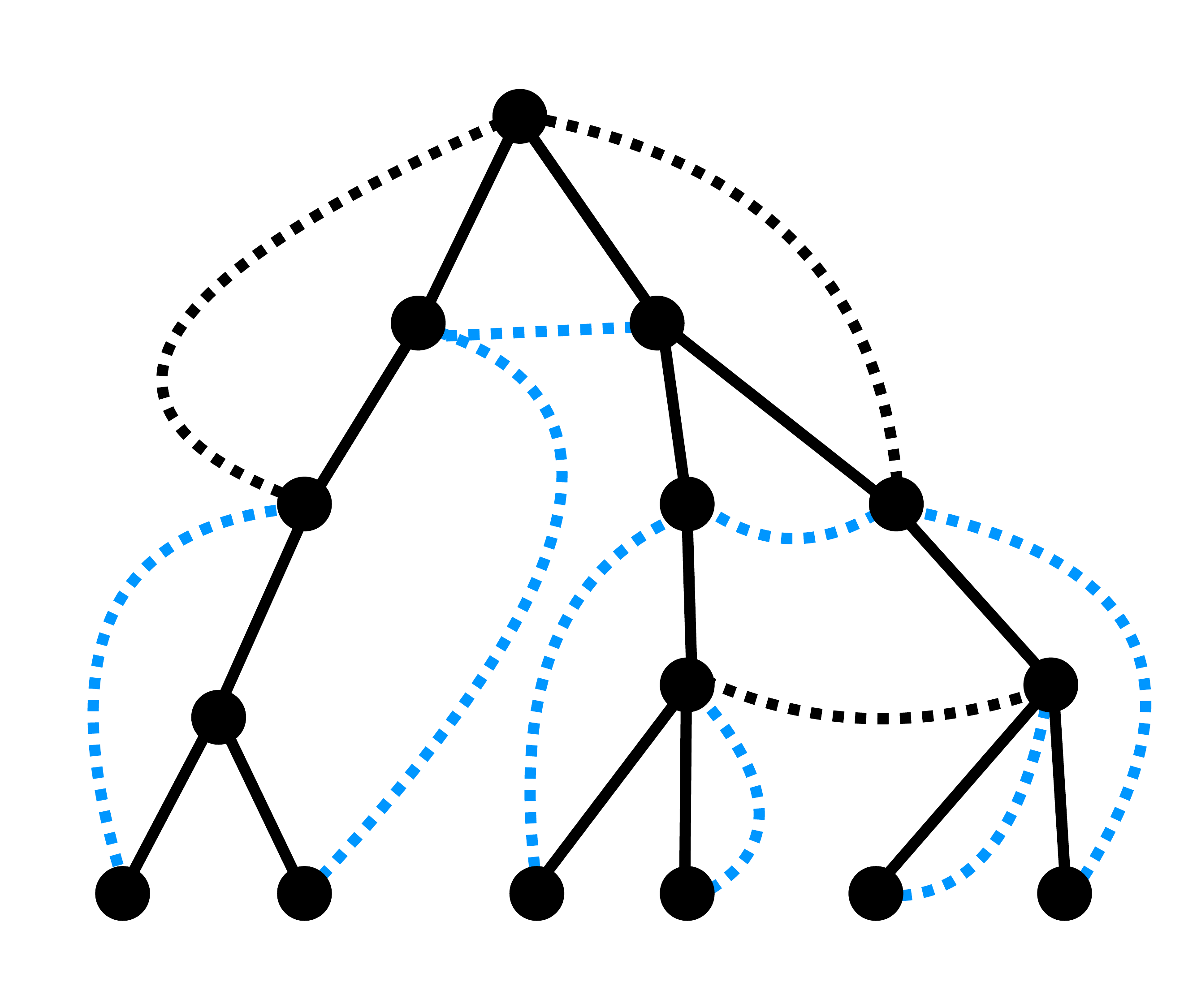}
		\caption{Feasible TAP solution.}\label{sfig:tap2}
	\end{subfigure}
	\hfill
	\begin{subfigure}[b]{.32\textwidth}
		\centering
		\includegraphics[width=\textwidth, trim=0mm 0mm 230mm 0mm, clip ]{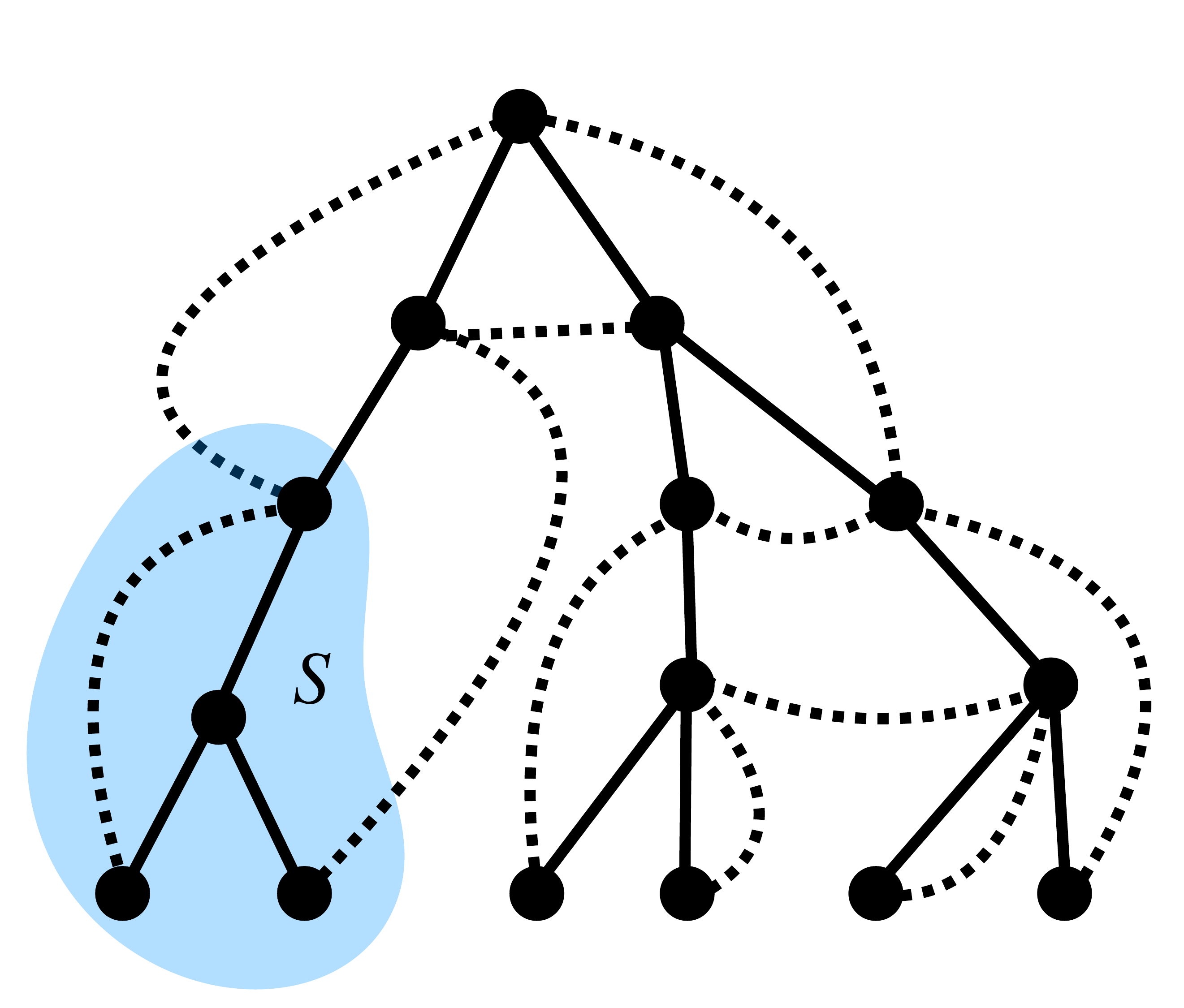}
		\caption{Tree cut $S$.}\label{sfig:tap3}
	\end{subfigure}
		 \begin{subfigure}[b]{.32\textwidth}
		\centering
		\includegraphics[width=\textwidth, trim=0mm 0mm 230mm 0mm, clip ]{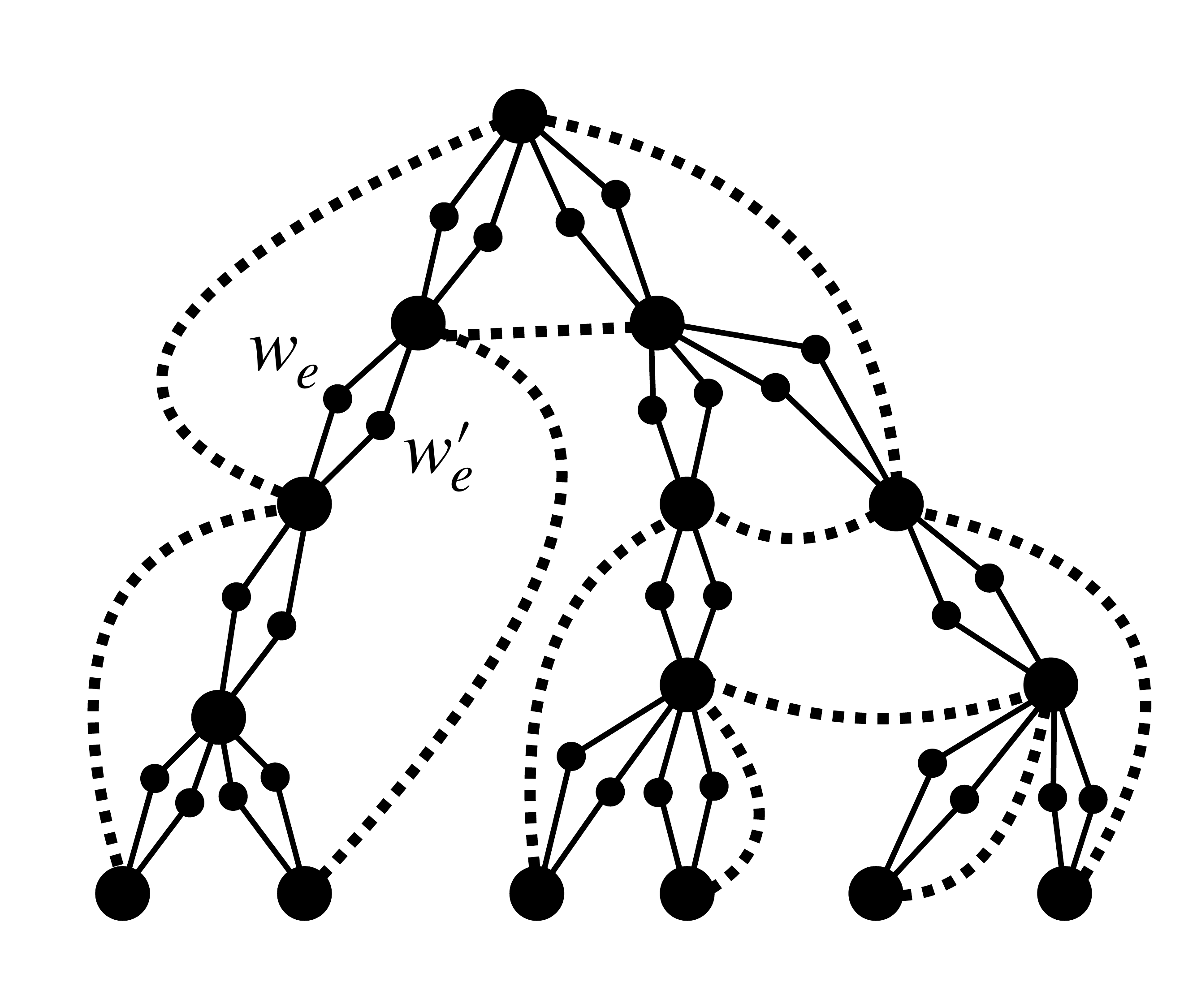}
		\caption{$k$-ECSM Instance on $H$.}\label{sfig:tap4}
	\end{subfigure}
	\hfill
	\begin{subfigure}[b]{.32\textwidth}
		\centering
		\includegraphics[width=\textwidth, trim=0mm 0mm 230mm 0mm, clip ]{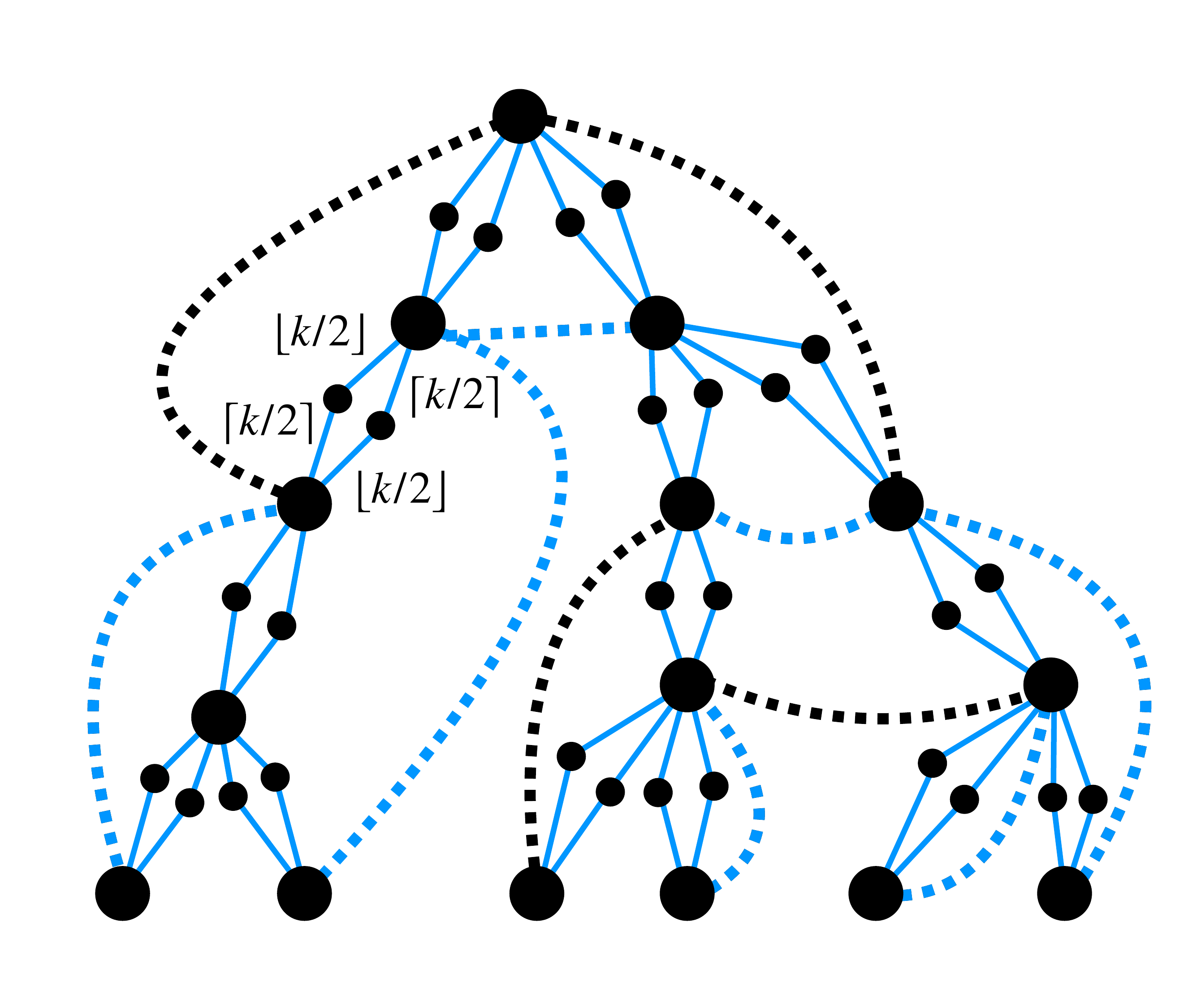}
		\caption{Feasible $k$-ECSM Solution.}\label{sfig:tap5}
	\end{subfigure}
	\hfill
	\begin{subfigure}[b]{.32\textwidth}
		\centering
		\includegraphics[width=\textwidth, trim=0mm 0mm 230mm 0mm, clip ]{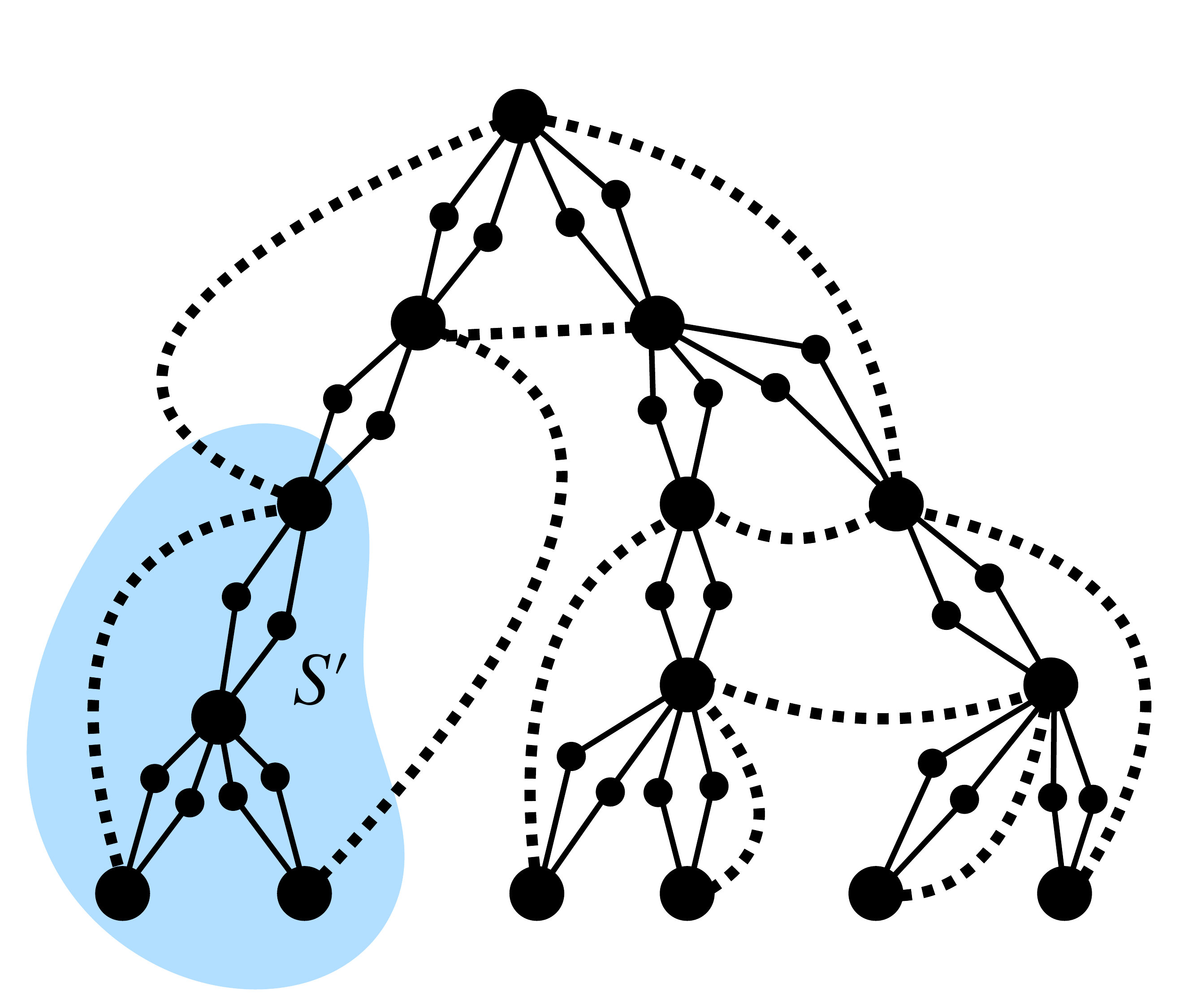}
		\caption{Corresponding cut $S'$.}\label{sfig:tap6}
	\end{subfigure}
	\caption{Our reduction from TAP to $k$-ECSM. \ref{sfig:tap1}: TAP instance on $G$ with links $L$ dashed and edge $e$ labeled. \ref{sfig:tap2}: feasible TAP solution in blue. \ref{sfig:tap3}: a tree cut $S$. \ref{sfig:tap4}: $k$-ECSM instance on $H$ with $w_e$, $w_e'$ labeled; $V_E$ are the small nodes. \ref{sfig:tap5}: feasible $k$-ECSM solution for $H$ where non-zero links in blue and values of edges incident to $w_e$ and $w_e'$ labeled. \ref{sfig:tap6}: cut $S'$ in $H$ corresponding to $S$ in $G$.}
\label{fig:oddReduction}
\end{figure}


A rough sketch of our analysis is as follows. By simple modifications to our $k$-ECSM solution, we can guarantee that for each vertex of $V_E$, one incident edge has value $\lfloor \sfrac{k}{2} \rfloor$ and one incident edge has value $\lceil \sfrac{k}{2} \rceil$. It follows that there must be some link in the support of our $k$-ECSM solution that covers $e$ (in the TAP sense), and so all links in the support of our $k$-ECSM solution form a feasible TAP solution. See \Cref{sfig:tap5}. On the other hand, one can show that the links in the support of our $k$-ECSM solution must make up at most an $O(\sfrac{1}{k})$ fraction of the total cost of our $k$-ECSM solution (on links and edges). Hence, if our $k$-ECSM solution is much better than $\left(1 + \frac{\epsTAP}{k}\right)$-approximate, the links in the support of our $k$-ECSM solution must be much better than $\left(1+\epsTAP\right)$-approximate as compared to the optimal TAP solution.

The following summarizes the simple modifications we will make to guarantee that the incident edges of each $w \in V_E$ have values $\lfloor k/2 \rfloor $ and $\lceil k/2 \rceil$.
\begin{lemma}\label{lem:rebalance}
  Given a feasible and integral solution $z'$ for the above instance of $k$-ECSM on $H$ for $k \in \mathbb{Z}_{\geq 1}$ odd, one can in poly-time compute a feasible and integral $z$ of equal cost such that for each $w \in V_E$, one of the two edges incident to $w$ has $z'$-value $\lfloor \sfrac{k}{2} \rfloor$ and one has $z'$-value $\lceil \sfrac{k}{2} \rceil$.
\end{lemma}	
\begin{proof}
  We initially set $z'_{\ell} = z_{\ell}$ for all $\ell\in L$.
  (The other entries of $z'$, i.e.,  $z'_{b}$ for $b\in \gadgetEdges$, will be set directly later and do not need to be initialized.)
  For each edge $e=\{u,v\}\in E$ of the TAP graph $G$ (in an arbitrary order), we do the following:
  \begin{itemize}
    \item For each $w\in \{w_e, w'_e\}$, we assign a $z'$-value of $\lceil \sfrac{k}{2} \rceil$ to one of the two edges incident to $w$ and a $z'$-value of $\lfloor \sfrac{k}{2} \rfloor$ to the other edge incident to $w$.
    \item We select an arbitrary link $\ell\in \cov(e)$ and increase its value by $z(\delta_H(w_e))+z(\delta_H(w_e'))-2k$.
      (Because $z(\delta_H(w_e)), z(\delta_H(w_e')) \geq k$, as $z$ is a $k$-ECSM solution, this increase is non-negative.)
  \end{itemize}
  Note that $z'$ can clearly be computed in poly-time, and it fulfills by construction that each vertex $w\in V_E$ is incident to one edge of value $\lceil \sfrac{k}{2}\rceil$ and one of value $\lfloor \sfrac{k}{2} \rfloor$.
  Hence, it remains to show that $z'$ is a $k$-ECSM solution.

  Any cut $Q \subseteq W$ such that $|\delta_H(Q)\cap \gadgetEdges| \geq 3$ has $z'$-value $z'(\delta_H(Q)) \geq 3 \lfloor \sfrac{k}{2}\rfloor \geq k$, as desired.
  The other cuts of $H$ contain precisely $2$ gadget edges, which must both be within the same cycle.%
\footnote{This is a well-known and easy to check property of cacti, and $(W,\gadgetEdges)$ is a cactus, i.e., a connected graph where every edge is in precisely one cycle.}

  Hence, let $Q\subseteq W$ be a such a cycle in $H$ that contains precisely two edges of $\gadgetEdges$, say within the cycle/gadget corresponding to $e\in E$.
  If either $\delta_H(w_e) \subseteq \delta_H(Q)$ or $\delta_H(w'_e) \subseteq \delta_H(Q)$, then we again obtain $z'(\delta_H(Q))\geq k$, because the two edges incident with $w_e$ (or $w'_e$) have $z'$-values that sum up to $k$.
  Thus, we can assume that $\delta_H(Q)$ contains one edge incident with $w_e$ and one incident with $w_e'$.
  These two edges already lead to a $z'$-value of at least $2 \lfloor \sfrac{k}{2}\rfloor = k-1$.
  
  We complete the proof by showing that $z'(\delta(Q)\cap L) \geq 1$.
  Note that $Q \cap V$ corresponds to a tree cut in $G$, with $e$ being the only edge of $E$ with one endpoint in $Q$ and one outside.
  Hence, $\delta(Q)\cap L = \cov(e)$, and it suffices to show that $z'(\cov(e)) \geq 1$.
  If $z(\cov(e))\geq 1$, then we are done because the $z'$-value on any link is no less than its $z$-value.
  Otherwise, we have $z(\cov(e)) = 0$, and observe that $z(\delta_H(w_e))+z(\delta_H(w_e')) > 2k$; indeed, if $z(\delta_H(w_e))=k$ and $z(\delta_H(w_e')) = k$, then there is an edge in both $\delta_H(w_e)$ and $\delta_H(w'_e)$ with $z$-value at most $\lfloor \sfrac{k}{2} \rfloor$, and the cut in $H$ that contains these two edges and the links in $\cov(e)$ has $z$-value strictly below $k$, which violates that $z$ is a $k$-ECSM solution.
  Hence, in our construction of $z'$, we increased the value of at least one link $\ell\in \cov(e)$ by $z(\delta_H(w_e)) + z(\delta_H(w'_e)) - 2k \geq 1$, which finishes the proof.
\end{proof}	

We conclude our hardness of approximation.
\hardness*
\begin{proof}

    As discussed above, we reduce from unweighted TAP. Suppose that we are given an instance of unweighted TAP on graph $G = (V,E)$ with links $L$ and optimal value $\OPTTAP \geq \frac{1}{4} \cdot |E|$. Assume for the sake of deriving a contradiction that one can always $\left(1+\frac{\epsTAP}{9k}\right)$-approximate $k$-ECSM and let $H$ be the instance of $k$-ECSM resulting from the above reduction. Let $z'$ be a $\left(1+\frac{\epsTAP}{9k}\right)$-approximate solution to the $k$-ECSM instance and let $z$ be the result of applying \Cref{lem:rebalance} to $z'$. Lastly, let $F$ be all links in the support of $z$.
Clearly constructing $F$ takes poly-time.

    We claim that $F$ is a feasible solution for TAP on $G$. Consider a tree cut $S \subseteq G$ cutting $e$ in $G$. Let $S'\subseteq W$ be the corresponding cut in $H=(W,B)$. By the guarantees of \Cref{lem:rebalance}, we know that some $\tilde{S} \in \{S', S'+w_{e}, S'+w_{e}', S'+w_{e}+w_{e}'\}$ satisfies $z(\delta(\tilde{S})\setminus L) = k -1$.
   (For example, in \Cref{sfig:tap6} it is $S' + w_{e}$.)
   However, since $z$ is feasible, it follows that some link incident to $\tilde{S}$ and therefore incident to $S$ must be in the support of $z$. Thus, $F$ is feasible.

    Lastly, we argue that $F$ has small cost, resulting in a contradiction to \Cref{thm:tapAPXHard}. Roughly, we observe that all gadget edges we add make up at most a $1 - \frac{1}{k}$ fraction of the total cost of our solution.
   So $(1 + \frac{\epsTAP}{9k})$-approximating $k$-ECSM translates to $(1 + \epsTAP)$-approximating TAP. More formally, let $\OPTkECSM$ be the value of the optimal $k$-ECSM solution to our problem on $G'$. It is easy to verify that the $k$-ECSM solution which for each $w \in V_e$ sets one incident edge to $\lfloor k/2 \rfloor$ and the other to $\lceil k/2 \rceil$, and sets each edge in the support of the optimal TAP solution to $1$ and all other edges to $0$, is feasible for $k$-ECSM. It follows that
    \begin{align}\label{eq:kECSMSoln}
        \OPTkECSM \leq \OPTTAP + 2k\cdot|E|.
    \end{align}
    On the other hand, observe that by the fact that $c^{\top} z' = c^{\top} z$ and the definition of $F$, we have
    \begin{align}\label{eq:zCost}
      c^{\top} z  = c^{\top} z' \geq |F| + 2k \cdot |E|.
    \end{align}

    Combining \Cref{eq:kECSMSoln}, \Cref{eq:zCost}, the fact that $c^{\top} z = c^{\top} z'$, and the fact that $z'$ is $(1+ \frac{\epsTAP}{9k})$-approximate by assumption, we have
     \begin{align*}
       \left(1+ \frac{\epsTAP}{9k}\right)(\OPTTAP + 2k\cdot |E|) \geq \left(1+ \frac{\epsTAP}{9k}\right)\OPTkECSM \geq c^{\top} z \geq |F| + 2k \cdot |E|.
    \end{align*}
    Rearranging the above, applying the fact that $\frac{1}{4} \cdot |E| \leq \OPTTAP$, and $k \geq 1$, we get
    \begin{equation*}
        |F| \leq \left( 1 + \frac{\epsTAP}{9k} \right) \cdot \OPTTAP + \frac{2\epsTAP}{9} \cdot |E|
            \leq \left( 1 + \epsTAP \right) \cdot \OPTTAP,
    \end{equation*}
    contradicting \Cref{thm:tapAPXHard}.
\end{proof}

\newpage

\printbibliography

\newpage

\appendix
\section{Results for Subset $k$-ECSM}\label{sec:steiner}
In this section, we observe that our $1 + O(\sfrac{1}{k})$-approximation for $k$-ECSM extends to the subset $k$-ECSM problem.

In subset $k$-ECSM we are given a multi-graph $G = (V, E)$, a subset of vertices $W \subseteq V$, an edge cost function $c:E \to \mathbb{R}_{\ge 0}$, and our goal is to find a set of edges $F \subseteq E$ where $W$ is $k$-edge-connected in $(V, F)$ while minimizing $c(F) \coloneqq \sum_{e\in F} c(e)$.\footnote{A subset $W \subseteq V$ is $k$-edge connected if and only if $|\delta(S) \cap F| \geq k$ for every $S \subseteq V$ such that $S \cap W \neq \emptyset$ and $W \setminus S \neq \emptyset$.} We denote the cost of the optimal solution with $\OPTkSECSM[k]$.

The following is the natural LP for subset $k$-ECSM.
\begin{equation}\label{lp:secsm}
\begin{array}{rrcll}
  \min &c^\top x     &     &                       &\\
       &x(\delta(S)) &\geq &k                      &\forall S\subseteq, S \cap W \neq \emptyset, W\setminus S \neq \emptyset \\
       &x            &\in  &\mathbb{R}^E_{\geq 0}. &
\end{array}\tag{$k\mathrm{-SECSM}~\mathrm{LP}$}
\end{equation}
\noindent We let \LPOPTkSECSM[k] give the optimal value of the above LP. \cite{Pri11} noted that the ``parsimonious property'' of \cite{GB93} implies that LP-competitive algorithms for $k$-ECSM give LP-competitive algorithms for subset $k$-ECSM.  We provide the details of this argument below.

For subset $k$-ECSM the parsimonious property allows one to assume that every vertex of $W$ has degree exactly equal to $k$ and every vertex not in $W$ has degree $0$ provided costs obey the triangle inequality.\footnote{Costs $c$ obey the triangle inequality if $c_{uv} \leq c_{uw} + c_{wv}$ for all $u,v,w \in V$.} The following summarizes this.
\begin{theorem}[Theorem 1 of \cite{GB93}, Applied to Subset $k$-ECSM]\label{thm:pars}
    Consider an instance of subset $k$-ECSM where $c$ satisfies the triangle inequality. Let \ref{lp:secsm}' be \ref{lp:secsm} for this instance but  with the added constraints 
    \begin{align}\label{eq:parsConstraints}
        x(\delta(v)) = 
        \begin{cases} 
            k & \text{if $v \in W$}\\
            0 & \text{otherwise}
        \end{cases}
    \end{align}
    for all $v \in V$. Then
    \begin{align*}
        \LPOPTkSECSM[k] = \LPOPTkSECSM[k]'.
    \end{align*}
\end{theorem}
The above parsimonious property almost immediately implies that algorithms for $k$-ECSM also give algorithms for subset $k$-ECSM.
\begin{lemma}[\cite{GB93, Pri11}]
    If there exists a poly-time algorithm for $k$-ECSM that returns a solution of cost at most $\alpha \cdot \LPOPTkECSM[k]$ then there is a poly-time algorithm for subset $k$-ECSM that returns a solution of cost at most $\alpha \cdot \LPOPTkSECSM[k]$.
\end{lemma}
\begin{proof}
    Consider an instance of subset $k$-ECSM on graph $G=(V,E)$ with subset $W \subseteq V$ and costs $c$. We let \LPOPTkSECSM[k] be the value of \ref{lp:secsm} for this instance (with costs $c$).
    
    Let $c'$ be new costs where for each $u,v \in V$ we let $c'_{uv}$ be the cost of the minimum cost (according to $c$) path from $u$ to $v$. Letting \LPOPTkSECSM[k]' be the value of \ref{lp:secsm} with costs $c'$ we trivially have
    \begin{align}\label{eq:a}
        \LPOPTkSECSM[k] = \LPOPTkSECSM[k]'.
    \end{align}

    Next, let \LPOPTkSECSM[k]'' be the value of \ref{lp:secsm} with costs $c'$ and the constraints of \Cref{eq:parsConstraints} added. By \Cref{thm:pars} (the parsimonious property) and the fact that $c'$ obeys the triangle inequality, we have that
    \begin{align}\label{eq:b}
        \LPOPTkSECSM[k]' = \LPOPTkSECSM[k]''.
    \end{align}

    Next, consider the instance of $k$-ECSM on graph $G = (W, E[W])$ with costs $c'$. Let $\LPOPTkECSM[k]$ be the value of \ref{lp:ecsm} for this instance. Since we know there is a solution to \ref{lp:secsm} with the constraints of \Cref{eq:parsConstraints} added of cost at most $\LPOPTkSECSM[k]''$ (according to $c'$), it follows that 
    \begin{align}\label{eq:c}
        \LPOPTkSECSM[k]'' \geq \LPOPTkECSM[k].
    \end{align}

    Lastly, observe that any feasible solution to $k$-ECSM on graph $G = (W, E[W])$ with costs $c'$ can, in poly-time, be turned into a feasible solution for subset $k$-ECSM with subset $W$ on $G$ with costs $c$ while not increasing the cost of the solution. It follows by Equations~\eqref{eq:a}, \eqref{eq:b}, and \eqref{eq:c} that if we can, in poly-time, compute a solution to $k$-ECSM on graph $G = (W, E[W])$ with costs $c'$ of cost at most $\alpha \cdot \LPOPTkECSM[k]$, then we can, in poly-time, compute a solution to subset $k$-ECSM on $G$ with costs $c$ of cost at most $\alpha \cdot \LPOPTkSECSM[k]$ as required.
\end{proof}

Lastly, combining the above with \Cref{thm:mainECSM} (our approximation algorithm for $k$-ECSM) gives the following result for subset $k$-ECSM.
\begin{theorem}
There is a poly-time algorithm for subset $k$-ECSM that, for any subset $k$-ECSM instance with $k \in \mathbb{Z}_{\geq 1}$, returns a feasible solution of cost at most $(1+\frac{10}{k}) \cdot \LPOPTkSECSM[k] \leq (1+\frac{10}{k}) \cdot \OPTkSECSM[k]$.
\end{theorem}

\end{document}